\documentclass{lmcs} %

\keywords{non-wellfounded proofs, recursive coalgebras, well-founded coalgebras,
familial functors, polynomial functors}

\usepackage{stackengine} %
\usepackage{amssymb} %
\usepackage{stmaryrd} %
\usepackage{bussproofs}
\usepackage[all]{xy} \SelectTips{cm}{}  %
\usepackage{quiver}

\usepackage{hyperref}
\def\equationautorefname~#1\null{Equation~(#1)\null}
\theoremstyle{plain} %

\newif\ifdraft\draftfalse

\ifdraft
\usepackage{todonotes}

\newcommand{\conf}[1]{\textcolor{red}{#1\%}}
\else
\usepackage[disable]{todonotes}

\newcommand{\conf}[1]{}
\fi

\newcommand{\Jdg}{{\mathrm{Jdg}}}
\newcommand{\OJdg}[1]{\mathrm{Jdg}^{#1}}
\newcommand{\fml}{\mathrm{fml}}

\newcommand{\op}{\mathrm{op}}
\newcommand{\Coalg}[1]{\mathrm{Coalg}(#1)}
\newcommand{\Alg}[1]{\mathrm{Alg}(#1)}
\newcommand{\ol}[1]{\overline{#1}}

\newcommand{\prem}{\mathrm{prem}}
\newcommand{\ccl}{\mathrm{ccl}}

\newcommand{\Set}{\mathbf{Set}}
\newcommand{\Cat}{\mathbf{Cat}}

\newcommand{\prog}{\mathrm{prog}}
\newcommand{\prffunc}[1]{\polyfunc{#1}}
\newcommand{\arfunc}{\mathrm{ar}}
\newcommand{\id}{\mathrm{id}}
\newcommand{\cat}[1]{\mathbb{#1}}
\newcommand{\nat}{\mathbb{N}}
\DeclareMathOperator*{\colim}{colim}
\newcommand{\Gopfib}[1]{q_{#1}}
\newcommand{\GopfibTotal}[1]{\int \! {#1}}
\newcommand{\Gfib}[1]{p_{#1}}
\newcommand{\GfibTotal}[1]{\int \! {#1}}
\newcommand{\tightwidehat}[1]{\stackon[-8pt]{#1}{\widehat{\phantom{#1}}}}
\newcommand{\Psh}[1]{\widehat{#1}}
\newcommand{\PshTight}[1]{\tightwidehat{#1}}
\newcommand{\PshNonSmash}[1]{\widehat{#1}}

\newcommand{\traceto}[1]{
  \mathrel{\raisebox{-0.55ex}{$\stackrel{#1}{\rightsquigarrow}$}}
}

\newcommand{\Poly}{\mathbf{VPoly}}
\newcommand{\polyfunc}[1]{F[{#1}]}
\newcommand{\node}[2]{#1  \! \in  \! #2}
\newcommand{\Sub}[1]{\mathrm{Sub}(#1)}
\newcommand{\sstar}{{\scalebox{0.6}{$\star$}}} %

\newcommand{\pullbackmark}[2]{\save ;p+<.8pc,0pc>:(0,-1)::%
(#1) *{\phantom{Z}} %
;p+(#2)-(0,0) **@{-}%
;p-(#1)+(0,0) *{\phantom{Z}} **@{-} \restore}
\newcommand{\rloop}[2][-]{\save \POS!R(.7) \ar@(ru,rd)^#1{#2} \restore}
\newcommand{\lloop}[2][-]{\save \POS!L(.7) \ar@(lu,ld)_#1{#2} \restore}
\newcommand{\uloop}[2][-]{\save \POS!U(.7) \ar@(lu,ru)^(.8){#2} \restore}

\stackMath
\newcommand\xxrightarrow[2][]{\mathrel{%
  \setbox2=\hbox{\stackon{\scriptstyle#1}{\scriptstyle#2}}%
  \stackunder[0pt]{%
    \xrightarrow{\makebox[\dimexpr\wd2\relax]{$\scriptstyle#2$}}%
  }{%
   \scriptstyle#1\,%
  }%
}}
\begin{document}

\title[Coalgebraic Non-Wellfounded Proofs]{Coalgebraic Non-Wellfounded Proofs: {\\}Recursiveness and GTC}

\author[M.~Kori]{Mayuko Kori}

\address{Research Institute for Mathematical Sciences, Kyoto University, Japan}	%
\email{mkori@kurims.kyoto-u.ac.jp}  %

\begin{abstract}
\noindent Non-wellfounded proof systems impose a global condition called the \emph{global trace condition} (GTC) on 
a derivation tree to ensure soundness. Providing a categorical characterisation 
of the GTC that guarantees soundness remains challenging due to the global, 
non-compositional nature of these conditions and the infinitary structure of 
non-wellfounded proofs. 
We develop a coalgebraic framework for non-wellfounded proof systems
where
derivation trees are modelled as coalgebras of generalised polynomial functors on presheaves.
Since the GTC is a constraint on infinite paths in derivation graphs, we 
employ graphs of coalgebras and formulate the GTC coalgebraically as a 
condition on these graphs. 
Soundness is then formulated as the existence of a unique coalgebra-to-algebra 
morphism from a coalgebra representing a derivation graph to an algebra specifying semantics. 

Within this framework,
we characterise the GTC via recursive coalgebras:
a coalgebra satisfies the GTC if and only if its ordinal-annotated image
is recursive for every ordinal bound.
For semantic algebras obtained from annotated algebras by the induced
right-adjoint lifting, this yields soundness: every proof admits a unique
coalgebra-to-algebra morphism in the original proof system.
We demonstrate the framework through a non-wellfounded proof system for
the modal $\mu$-calculus and a non-wellfounded variant of Santocanale's
circular proof system in $\mu$-bicomplete categories.
\end{abstract}

\maketitle

\section*{Introduction}
In recent years,
non-wellfounded proof systems, including cyclic proofs, have gained attention as proof systems for logics with (co)induction.
Unlike ordinary finitary proofs whose derivation trees are finite,
non-wellfounded proofs may have infinite derivation trees.
As a consequence, soundness can no longer be established by induction on the depth of derivation trees. To ensure soundness, these proof systems 
impose an additional global condition on derivation trees,
known as the \emph{global trace condition} (GTC). 

This increasing interest has also led to the development of various categorical approaches to non-wellfounded and cyclic proofs.
Fortier and Santocanale
developed a categorical semantics for cyclic proofs 
in the setting of $\mu$-bicomplete categories~\cite{DBLP:conf/csl/FortierS13}.
More recently, Ehrhard et al.~proposed a categorical semantics for non-wellfounded and cyclic proofs in linear logic with fixed points~\cite{DBLP:conf/lics/EhrhardJS25}.
From a different perspective,
Afshari and Wehr introduced abstract cyclic proofs~\cite{AfshariW22}, in which the GTC is described categorically using trace categories.
Their work focuses mainly on the efficient algorithmic checking of the GTC,
rather than on soundness.

While various categorical approaches have been developed,
it still remains a challenge to
provide
a categorical characterisation of the GTC that ensures soundness.
One difficulty is that
the GTC imposes constraints on entire derivation trees,
which do not readily
align with the local and compositional nature of categorical reasoning.
In addition, 
non-wellfounded proofs may be 
infinite derivation trees, preventing the use of the standard approach in categorical semantics of constructing morphisms from proofs by structural induction on the derivation trees.

The aim of this paper 
is to develop
a categorical abstraction of non-wellfounded proofs equipped with a soundness result.
Besides defining traces explicitly over paths and formulating progress conditions operationally, 
we characterise the GTC as a categorical property---specifically, as recursiveness of coalgebras along suitable adjunctions---that guarantees the existence and uniqueness of semantic interpretations.

Since the GTC is a constraint on infinite paths in derivation trees,
we employ a coalgebraic approach.
We model derivation trees as coalgebras and exploit their associated graphs,
building on the notion of graphs of coalgebras~\cite{DBLP:conf/calco/Gumm05,DBLP:books/daglib/0031002,AdamekMM20}.
Specifically,
we work with
coalgebras of \emph{familial functors}
on presheaves.
These functors are generalisations of polynomial functors,
and allow us
to model derivation trees whose nodes are labelled by judgements and whose branching structure is determined by a fixed collection of inference rules.
We define the GTC directly for such coalgebras,
formulating as a condition on infinite paths in their associated derivation graphs.

In our framework,
a (possibly infinite) derivation tree is represented as a coalgebra $C \to FC$,
while
a semantics representing how to interpret inference rules is specified by an algebra $F\Omega \to \Omega$.
An interpretation of a proof is then expressed as 
a coalgebra-to-algebra morphism  between them, which assigns to each node of the derivation tree a semantic value 
in a manner that is consistent with the interpretation of inference rules.

\begin{tikzpicture}[>=stealth]
\node (C) at (2.2,1.2) {$C$};
\node (FC) at (2.2,-0.2) {$FC$};
\node (Omega) at (4,1.2) {$\Omega$};
\node (FOmega) at (4,-0.2) {$F\Omega$};
\node (where) at (9, 0.5) {(in $\Psh{\Jdg}$)};

\draw[->] (C) -- (FC) node[midway, align=right, left] {representing\\a derivation tree};
\draw[->] (FOmega) -- (Omega)  node[midway, align=left, right]{representing\\a semantics of rules};
\draw[->, dotted] (C) -- (Omega) node[midway, above] {$\llbracket - \rrbracket$};
\draw[->, dotted] (FC) -- (FOmega)  node[midway, above] {$F\llbracket - \rrbracket$};

\end{tikzpicture}

Soundness is then formulated as the existence of a unique coalgebra-to-algebra morphism.
To guarantee such existence,
\emph{well-founded} or \emph{recursive coalgebras} have been well studied~\cite{DBLP:journals/ita/AdamekLM07,AdamekMM20,DBLP:journals/mscs/JeanninKS17}.
They provide a witness explaining why solutions of recursive definitions are unique,
and successfully capture the semantics of recursive functions when the underlying coalgebras are well-founded.

However, derivation trees of non-wellfounded proofs may be infinite and hence
give rise to non-wellfounded coalgebras.
It has been observed that 
unique solutions can still exist in certain non-wellfounded cases,
and existing theories of recursive or well-founded coalgebras do not explain this phenomenon
at least in a straightforward way~\cite{capretta07,DBLP:journals/mscs/JeanninKS17}.

To address this gap, 
we exploit the correspondence of coalgebra-to-algebra morphisms along adjunctions given in ~\cite{DBLP:conf/popl/HinzeWG15}. 
Instead of requiring the original coalgebra to be recursive, we study whether its image under a suitable left adjoint is recursive.
Intuitively, this amounts to enriching the original derivation graph with an additional dimension---obtained by pulling back along an adjunction---so that
the expanded graph becomes well-founded, even when the original graph is not.

Under an assumption on the semantic algebra,
we show a soundness theorem:
if a coalgebra representing a derivation graph satisfies the GTC, then there exists a unique coalgebra-to-algebra morphism.
In the context of concrete non-wellfounded proofs, 
enriching the derivation graph
corresponds to ordinal annotations on fixed-point operators, a technique commonly employed in concrete soundness proofs~\cite{DBLP:journals/tcs/NiwinskiW96,DBLP:journals/logcom/BrotherstonS11,DBLP:conf/csl/KoriT021},
and
the assumption on the semantic algebra means that
the semantics of fixed-point operators are given by ordinal-indexed iterative computations.

We further investigate the relationship between the GTC and recursiveness of coalgebras.
In particular, we establish that
the GTC and recursiveness characterise each other:
we give
a characterisation of the GTC in terms of recursive coalgebras,
and conversely
a characterisation of recursiveness in terms of the GTC.
We also show that
the GTC is preserved under certain functors given by right adjoints,
analogous to a known result for recursive coalgebras.

Taken together,
the main conceptual message of this paper is twofold.  
First, the GTC can be understood as the condition ensuring that, after transport along a suitable adjunction, the relevant coalgebra becomes recursive (equivalently, well-founded in the associated graph). 
Secondly, in concrete examples,
the suitable adjunction is guided by the semantics.
More precisely, the ordinal-indexed approximants 
used to interpret
fixed-point operators provide a lifting of the semantic algebra along
this adjunction.

The contributions of this paper are as follows:
\begin{itemize}
  \item We extend the notion of graphs of coalgebras~\cite{DBLP:conf/calco/Gumm05,DBLP:books/daglib/0031002,AdamekMM20}
  from endofunctors on $\Set$ to 
  those on presheaves. We also establish the equivalence 
  between well-foundedness of graphs and recursiveness (or well-foundedness) of coalgebras (\autoref{thm:wf}) as an extension of their result.
  \item We propose an approach to ensure existence and uniqueness of coalgebra-to-algebra morphisms in non-wellfounded cases (\autoref{prop:rec_solution}). 
  This approach provides a coalgebraic explanation of a non-wellfounded example appearing in \cite{capretta07,DBLP:journals/mscs/JeanninKS17}, that admits a unique solution
    and is not covered by existing frameworks for recursive coalgebras.
  \item We propose a coalgebraic abstraction of non-wellfounded proof systems via coalgebras of familial functors on presheaves,
  and establish a soundness result within this framework (\autoref{thm:soundness}).
  \item We provide a characterisation of the GTC in terms of recursive coalgebras, and examine the relationship between recursiveness and the GTC (\autoref{prop:rec_GTC}, \autoref{thm:GTC_rec}).
  \item We illustrate our framework with two proof systems: 
  a non-wellfounded proof system for the modal $\mu$-calculus~\cite{DBLP:conf/lics/AfshariL17} and a non-wellfounded variant of Santocanale's circular proof system~\cite{DBLP:conf/fossacs/Santocanale02}.
\end{itemize}

\paragraph*{Structure of the paper.}
We first develop technical machinery in \autoref{sec:graph}--\autoref{sec:steps}.
After recalling basic notions on recursive and well-founded coalgebras
in \autoref{sec:preliminary}, 
we introduce graphs of coalgebras for endofunctors on presheaves in \autoref{sec:graph}.
In
\autoref{sec:coalg_alg_adj}, we review coalgebra-to-algebra morphisms in categories connected by adjunctions
and study
the unique existence of these morphisms, providing a coalgebraic explanation of a non-wellfounded example from~\cite{capretta07,DBLP:journals/mscs/JeanninKS17}.
Then, in 
\autoref{sec:steps},
we
introduce familial functors on presheaves.
Our abstract framework for non-wellfounded proof systems and associated results are presented in 
\autoref{sec:proof_system} and \autoref{sec:base_change}.
We then demonstrate our framework through concrete non-wellfounded proof systems in \autoref{sec:example}.  
In \autoref{sec:relatedwork}, we discuss related work, and \autoref{sec:conclusion} concludes the paper.

\section{Recursive and Well-Founded Coalgebras}

\subsection{Preliminaries} \label{sec:preliminary}

  Let $F\colon \cat{C} \to \cat{C}$ be a functor.
  For~$F$-coalgebras $c_1\colon C_1 \to FC_1$ and $c_2\colon C_2 \to FC_2$,
  a \emph{coalgebra morphism} $h\colon c_1 \to c_2$ is a morphism $h\colon C_1 \to C_2$ in $\cat{C}$ such that $Fh \circ c_1 = c_2 \circ h$.
  Dually, for~$F$-algebras $a_1\colon FA_1 \to A_1$ and $a_2\colon FA_2 \to A_2$,
  an \emph{algebra morphism} $h\colon a_1 \to a_2$ is a morphism $h\colon A_1 \to A_2$ in $\cat{C}$ such that $h \circ a_1 = a_2 \circ Fh$.
  We write $\Coalg{F}$ for the category of coalgebras and coalgebra morphisms, and $\Alg{F}$ for that of algebras and algebra morphisms.
\begin{defi}[ca-morphism]
  Let $F\colon \cat{C} \to \cat{C}$ be a functor.
  For an~$F$-coalgebra $c\colon C \to FC$ and an~$F$-algebra $a\colon FA \to A$,
 a \emph{ca-morphism} (short for \emph{coalgebra-to-algebra morphism}) from~$c$ to~$a$ is a morphism $i\colon C \to A$ in $\cat{C}$ such that 
 the following diagram commutes.
 \[
   \xymatrix{
     C \ar[r]^{i} \ar[d]_{c} & A \\
     FC \ar[r]_{Fi} & FA  \ar[u]_{a}
   }
 \]
\end{defi}

We introduce two important notions of coalgebras, \emph{recursive coalgebras} and \emph{well-founded coalgebras}, that capture well-founded structures.
These notions were first studied by Osius for the powerset functor~\cite{OSIUS197479},
and later generalised to general endofunctors by Taylor~\cite{Taylor21,DBLP:books/daglib/0031002}.
We refer to \cite{CaprettaUV06,AdamekMM20} for more details.
\begin{defi}[{recursive coalgebra}]
  Let $F\colon \cat{C} \to \cat{C}$ be a functor.
  A coalgebra $c\colon C \to FC$ is \emph{recursive}
  if for each~$F$-algebra $a\colon FA \to A$,
  there exists a unique ca-morphism from~$c$ to~$a$.
\end{defi}

\begin{defi}[{Well-founded coalgebra~\cite{DBLP:journals/corr/abs-1305-0576}}]
  Let $F\colon \cat{C} \to \cat{C}$ be a functor.
  A coalgebra $c\colon C \to FC$ is \emph{well-founded}
  if
  every subobject $i\colon A \rightarrowtail C$ for which there exists
  a morphism $A\to FA$ making the following square a pullback
  is an isomorphism.
  \[
    \xymatrix{
       *+<8pt>{A} \pullbackmark{0, 1}{1, 0}\ar@{>->}[d]^i \ar[r] &FA \ar[d]^{Fi} \\
      C \ar[r]^{c} &FC
    }
  \]
\end{defi}

\begin{exa} \label{eg:powerset}
  A coalgebra $c\colon C \to \mathcal{P}C$ of the powerset functor $\mathcal{P}\colon \Set \to \Set$
  represents a directed graph whose nodes are elements of~$C$ and whose edges are given by $\{(x, y) \mid y \in c(x)\}$.
  It is well-known that the following are equivalent:
  \begin{enumerate}
    \item~$c$ is recursive.
    \item~$c$ is well-founded.
    \item The graph is well-founded, i.e.~it contains no infinite path.
  \end{enumerate}
\end{exa}

\subsection{Graphs of Coalgebras in Presheaves} \label{sec:graph}
  A general result instantiating the equivalence discussed in~\autoref{eg:powerset}
  is known for endofunctors on $\Set$~\cite{DBLP:conf/calco/Gumm05,DBLP:books/daglib/0031002,AdamekMM20}.
  We show that
  the result extends naturally to endofunctors on presheaves.

\begin{nota}
For a functor $f\colon \cat{A} \to \Cat$,
we write
$\Gopfib{f}\colon \GopfibTotal{f} \to \cat{A}$
for the opfibration obtained
by the Grothendieck construction.
Thus, an object of $\GopfibTotal{f}$ is a pair $(a, x)$ where $a \in \cat{A}$ and $x \in f(a)$, and a morphism $(a, x) \to (a', x')$ is a pair consisting of a morphism $h\colon a \to a'$ in $\cat{A}$ and a morphism $f(h)(x) \to x'$ in $f(a')$.
Dually, 
for a functor $g\colon \cat{B}^\op \to \Cat$,
we write $\Gfib{g}\colon \GfibTotal{g} \to \cat{B}$ for the fibration obtained by the Grothendieck construction.
When applying the Grothendieck construction
to a $\Set$-valued functor, we regard sets as discrete categories,
i.e.~we use the inclusion \(\Set\hookrightarrow\Cat\).
We also use the notation 
$\node{x}{f(a)}$ and $\node{y}{g(b)}$ for the corresponding objects $(a, x) \in \GopfibTotal{f}$ and $(b, y) \in \GfibTotal{g}$.
We use $\iota_0, \iota_1$ and $\pi_0, \pi_1$ for the coproduct injections and product projections, respectively.
\end{nota}

For a small category $\cat{C}$, we write $\Psh{\cat{C}}$ for the presheaf category $\Set^{\cat{C}^\op}$.
Recall that an \emph{intersection} is a wide pullback (i.e.~an (arbitrary) product in a slice category) of monomorphisms, and presheaves have all intersections.
For a monomorphism $m\colon A \to B$ and a morphism $f\colon X \to B$, 
the \emph{inverse image} of~$m$ under~$f$ is the monomorphism given by the pullback of~$m$ along~$f$.
\newcommand{\Succ}[2]{\mathrm{Succ}_{#1}^{#2}}
\newcommand{\SuccOmit}[1]{\mathrm{Succ}_{#1}}
\begin{lem} \label{lem:graph_a}
  Let $F\colon \Psh{\cat{C}} \to \Psh{\cat{C}}$ be a functor and $c\colon C \to F C$ be a coalgebra.
  For each $(x, n) \in \GfibTotal{C}$ (i.e.~$x \in \cat{C}$ and $n \in C(x)$),
  we define the subpresheaf $\Succ{(x, n)}{c}$  of~$C$ by
  \[
    \Succ{(x, n)}{c} \coloneqq \bigcap \{B \mid B \hookrightarrow C, c_x(n) \in F(B)(x) \}.
  \]
  Then for each morphism $f\colon x \to y$ in $\cat{C}$ and $m \in C(y)$,
  \[\Succ{(x, Cf(m))}{c} \hookrightarrow \Succ{(y, m)}{c}.\]
  In other words, $\Succ{}{c}$ forms a functor $\GfibTotal{C} \to \Sub{C}$
  where $\Sub{C}$ is
  the category whose objects are subpresheaves of~$C$ and whose morphisms $B \to B'$ are 
  monomorphisms  making the inclusions $B \hookrightarrow C$ and $B' \hookrightarrow C$ commute.
\end{lem}
\begin{proof}
Let $i_B\colon B\hookrightarrow C$ be a subpresheaf.
For each $f\colon x\to y$ in $\cat{C}$, naturality gives the
commutative diagram
\[
\xymatrix{
  C(y) \ar[r]^-{c_y} \ar[d]_{C(f)} &
  (FC)(y) \ar[d]_{(FC)(f)} &
  (FB)(y) \ar[l]_-{(Fi_B)_y} \ar[d]^{(FB)(f)} \\
  C(x) \ar[r]_-{c_x} & (FC)(x) & (FB)(x) \ar[l]^-{(Fi_B)_x}.
}
\]
Hence $c_y(m)\in (FB)(y)$ implies
$c_x(C(f)(m))\in (FB)(x)$.
Note that a morphism $(x,n)\to(y,m)$ in $\GfibTotal{C}$ is precisely a morphism
$f\colon x\to y$ such that $C(f)(m)=n$, so these inclusions define the functor.
\end{proof}
  We omit the superscript~$c$ of $\Succ{(x, n)}{c}$ when it is clear from the context.
  This functor $\Succ{}{c}$ 
specifies the next nodes in the graph of~$c$, as in the following definition.

\begin{defi}[graph of a coalgebra] \label{def:graph}
  Let $F\colon \Psh{\cat{C}} \to \Psh{\cat{C}}$ be a functor and $c\colon C \to F C$ be a coalgebra.
  The \emph{graph of~$c$} is the directed graph 
  whose nodes are
  objects $(x, n) \in \GfibTotal{C}$ (i.e.~$x \in \cat{C}$ and $n \in C(x)$),
  and whose edges $(x, n) \to (y, m)$ are defined by the condition 
  $(y, m) \in \int \Succ{(x, n)}{c}$.
\end{defi}

\begin{lem} \label{lem:graph}
  Let $c\colon C \to F C$ be a coalgebra of a functor $F\colon \Psh{\cat{C}} \to \Psh{\cat{C}}$.
     If~$F$ preserves intersections, then $c_x(n) \in F(\Succ{(x, n)}{c})(x)$ holds
  for each node $(x, n) \in \GfibTotal{C}$.
\end{lem}
\begin{proof}
  Since $F$ preserves intersections, 
  $F(\Succ{(x,n)}{c})$ is isomorphic to $\bigcap \{FB \mid B \hookrightarrow C, c_x(n) \in F(B)(x) \}$.
  Intersections of presheaves are computed pointwise, and $c_x(n)$
  belongs to every one of these subobjects at $x$.
  It therefore belongs to their intersection at $x$.
\end{proof}

The following is an extension of \cite[Cor.~4.10]{AdamekMM20}.
\begin{thm} \label{thm:wf}
  Let $F\colon \Psh{\cat{C}} \to \Psh{\cat{C}}$ be a functor preserving intersections and inverse images,
  and
  $c\colon C \to F C$ be a coalgebra.
  Then the following statements are equivalent:
  \begin{enumerate}
    \item \label{item:rec}~$c$ is recursive.
    \item \label{item:wf}~$c$ is well-founded.
    \item \label{item:graph} The graph of~$c$ is well-founded, i.e.~it contains no infinite path.
  \end{enumerate}
\end{thm}
\begin{proof}
  $(\ref{item:rec} \Leftrightarrow \ref{item:wf})$
  This follows from~\cite[Thm.~5.1 and 5.8]{AdamekMM20}.

  $(\ref{item:wf} \Rightarrow \ref{item:graph})$
  Define a subpresheaf $B \hookrightarrow C$ by 
  \[
  B(x) \coloneqq \{n \mid \text{there is no infinite path from } (x, n) \text{ in the graph of }c\}.
  \]
  Let us show that it is natural.
  For each $f\colon x \to y$ in $\cat{C}$  and 
  $n \in B(y)$,
  if there is an edge $(x, Cf(n)) \to v$, 
  then there is also
  $(y, n) \to v$
  by \autoref{lem:graph_a}.
  Thus we have $Cf(n) \in B(x)$.

  For each $n \in B(x)$,
  it follows that $\SuccOmit{(x, n)} \hookrightarrow B$
  by definition of~$B$.
  Then \autoref{lem:graph} leads to $c_x(n) \in FB(x)$.
  Therefore,~$B$ forms a subcoalgebra of~$c$.
  Moreover, for each $x \in \cat{C}$ and $n \in C(x)$,
  $c_x(n) \in FB(x)$ implies 
  $\SuccOmit{(x, n)} \hookrightarrow B$ by definition of $\SuccOmit{(x, n)}$.
  It follows that all child nodes~$v$ of $(x, n)$ are well-founded (i.e.~there is no infinite path from~$v$),
  and thus
  $n \in B(x)$. 
  Therefore, 
  the subcoalgebra given by~$B$ is cartesian,
  that is,
  the following diagram forms a pullback.
  \[
    \xymatrix{
      *+<8pt>{B} \pullbackmark{0, 1}{1, 0}\ar@{^{(}->}[d] \ar[r] &*+<8pt>{FB} \ar@{^{(}->}[d] \\
      C \ar[r]^{c} &FC
    }
  \]
  Because~$c$ is well-founded, we have $B = C$.

  ($\ref{item:graph} \Rightarrow \ref{item:wf}$)
  Assume
  that the graph of~$c$ is well-founded and
  there is a cartesian subcoalgebra $b\colon B \to FB$ of~$c$ such that $B \neq C$.
  Then there is $n_0 \in  C(x_0) \setminus B(x_0)$ for some $x_0 \in \cat{C}$.
  Since~$b$ is cartesian,
  $c_{x_0}(n_0) \not \in FB(x_0)$. 
  This shows that $\SuccOmit{(x_0, n_0)} \not \hookrightarrow B$ because
  $\SuccOmit{(x_0, n_0)} \hookrightarrow B$ implies $c_{x_0}(n_0) \in F(\SuccOmit{(x_0, n_0)})(x_0) \subseteq FB(x_0)$ by \autoref{lem:graph}.
  Therefore, there is $n_1 \in C(x_1) \setminus B(x_1)$ for some $x_1 \in \cat{C}$ such that $n_1 \in \SuccOmit{(x_0, n_0)}(x_1)$, and thus there is an edge from $(x_0, n_0)$ to $(x_1,  n_1)$.
  Repeating this, we obtain an infinite path $\langle(x_i, n_i)\rangle_{i \in \nat}$ in the graph of~$c$.
  This is a contradiction.
\end{proof}
\section{Correspondence of Ca-morphisms} \label{sec:coalg_alg_adj}
This section develops our approach to establishing the existence of unique ca-morphisms in non-wellfounded settings. 
The key idea is to exploit adjunctions:
rather than requiring a coalgebra to be recursive, we study whether its image under a suitable left adjoint is recursive.
This provides a framework that accommodates non-wellfounded cases which do not fit straightforwardly into existing theories of recursive coalgebras.
We illustrate this approach with a concrete example in \autoref{sec:desc}.

\subsection{Steps and Ca-morphisms}
We begin by recalling the correspondence of ca-morphisms along an adjunction given in~\cite[Thm.~3.4]{DBLP:conf/popl/HinzeWG15}.
\begin{prop} \label{prop:steps}
  Let $L \dashv R\colon \cat{D} \to \cat{C}$ be an adjunction,
  and $F\colon \cat{C} \to \cat{C}$ and
  $G\colon \cat{D} \to \cat{D}$ be functors.
  Then there is a bijective correspondence
  between natural transformations $\tau\colon LF \Rightarrow GL$ and $\sigma\colon FR \Rightarrow RG$, given by the so-called mate bijection.
  Moreover,
  such natural transformations induce 
  the liftings 
  \[
  \Alg{R}\colon \Alg{G} \to \Alg{F}, \ \Coalg{L}\colon \Coalg{F} \to \Coalg{G}
  \]
  of the functors~$R$ and~$L$ to the category of algebras and coalgebras, respectively,
  defined by $\Alg{R}(a) = R(a) \circ \sigma$, $\Coalg{L}(c) = \tau \circ L(c)$
  for each~$G$-algebra~$a$ and~$F$-coalgebra~$c$.
\end{prop}
A natural transformation $\tau$ (or $\sigma$) as above is called a \emph{step}.
A fundamental property of steps is that they allow us to
relate ca-morphisms along the adjunction $L \dashv R$, as follows.

\begin{lem} \label{lem:coalg_to_alg}
  Consider the setting of \autoref{prop:steps} with a step $\tau\colon LF \Rightarrow GL$.
  For each~$F$-coalgebra $c\colon C \to FC$ and~$G$-algebra $a\colon GA \to A$,
  there is a bijective correspondence between ca-morphisms from~$c$ to $\Alg{R}(a)$
  and those from $\Coalg{L}(c)$ to~$a$.
  This correspondence is given by the adjunction $L \dashv R$.
\end{lem}
An important consequence of this correspondence is 
that the lifting of a left adjoint preserves recursive coalgebras,
in direct analogy with the basic fact that a left adjoint preserves colimits.
\begin{prop}[{\cite[Prop.~12]{CaprettaUV06}}] \label{prop:rec_coalg}
  Consider the setting of \autoref{prop:steps} with a step $\tau\colon LF \Rightarrow GL$.
  For any recursive~$F$-coalgebra~$c$,
  the~$G$-coalgebra $\Coalg{L}(c)$ is also recursive.
\end{prop}
This result has been exploited in the study of coalgebraic trace semantics~\cite{RotJL21}.
There, 
it provides a way to obtain a unique ca-morphism: 
if an~$F$-coalgebra~$c$ is recursive,
then there is a unique ca-morphism
from $\Coalg{L}(c)$ to a given algebra.

In this paper, we investigate a different perspective.
We start with an~$F$-coalgebra~$c$ that is not necessarily recursive;
however, its image $\Coalg{L}(c)$ is recursive.
Under an appropriate condition on the target algebra, this still guarantees a unique ca-morphism from~$c$.
This is formalised in the following proposition, which directly follows from \autoref{lem:coalg_to_alg}.

\begin{prop} \label{prop:rec_solution}
  Consider the setting of \autoref{prop:steps} with a step $\tau\colon LF \Rightarrow GL$.
  Let~$c$ be an~$F$-coalgebra and~$a$ be an~$F$-algebra.
    Suppose that
    there exists a~$G$-algebra $a'$ such that $\Alg{R}(a') \cong a$.
    Then there exists a unique ca-morphism from~$c$ to~$a$
    if
    $\Coalg{L}(c)$ is recursive.
\end{prop}
This provides an abstract mechanism underlying soundness of non-wellfounded proofs.
Rather than requiring a coalgebra to be recursive directly, it suffices that recursiveness holds after transporting the coalgebra along a left adjoint.
This observation will later allow us to interpret global trace conditions as recursiveness conditions on enriched derivation graphs.

By \autoref{prop:rec_coalg},
the condition that $\Coalg{L}(c)$ is recursive
is weaker than the condition that~$c$ itself is recursive.
This can be understood as a restriction of the universal property:
$\Coalg{L}(c)$ is recursive if and only if
there is a unique ca-morphism from~$c$ 
to each $F$-algebra in the image of $\Alg{R}$,
rather than to 
all~$F$-algebras.
This partial recursiveness condition allows us to reason about non-wellfounded cases, as illustrated in the next subsection.

\subsection{Example: Descending Sequence} \label{sec:desc}
We illustrate \autoref{prop:rec_solution} with a non-wellfounded example,
originally 
presented in Capretta's talk~\cite{capretta07} and later discussed in Jeannin et al.'s paper~\cite[Sec.~5.1]{DBLP:journals/mscs/JeanninKS17}.
We later see that this example is a toy instance of a B\"{u}chi-style GTC.

\subsubsection{Setup}
\newcommand{\descending}{\mathrm{ds}}
\newcommand{\cons}[2]{#1\!:\!#2}
Let $\nat$ be the set of natural numbers,
and $\omega$ be the posetal category of natural numbers under the usual order.
We write 
$\nat^\nat$ for the set of infinite streams of natural numbers,
and
$\cons{i}{t}$ for the stream with head $i \in \nat$ and tail $t \in \nat^\nat$.

Consider a function $\descending\colon \nat \times \nat^\nat \to \nat^\nat$ 
satisfying the following:
\begin{align}  \label{eq:h}
  \descending(n, \cons{i}{\cons{j}{xs}}) = 
  \begin{cases}
\descending(n+1, \cons{j}{xs}) &\text{ if }i > j \\ 
\cons{n}{\descending(1, \cons{j}{xs})} &\text{ if }i \leq j.
  \end{cases}
\end{align}
This function is actually uniquely determined. 
For any stream $xs$,
the output $\descending(1, xs)$ records 
the lengths of consecutive descending sequences from the beginning of $xs$.
For instance, $\descending(1, 4\colon\!2\colon\!3\colon\!7\colon\!6\colon\!5\colon\!9\colon\!\cdots) = (2\colon\!1\colon\!3\colon\! \cdots)$, counting the lengths of descending sequences $4>2$, $3$, $7>6>5$, and so on.

\subsubsection{Why is the function uniquely determined?}
Three key properties explain unique existence of a function $\descending$:
\begin{enumerate}
  \item[(i)]
 The output stream is determined by all of its finite prefixes.
\item[(ii)]
 Each time the case $i \leq j$ in \autoref{eq:h} is taken during a call sequence,
it extends the determined finite prefix by one element.
\item[(iii)] The case $i \leq j$ is taken infinitely often in any infinite call sequence.
\end{enumerate}
We aim to reformulate this in terms of recursive coalgebras.

\subsubsection{Coalgebraic reformulation}
\newcommand{\Fds}{F_\descending}
\newcommand{\cds}{c_\descending}
\newcommand{\ads}{a_\descending}
Any function $\descending\colon \nat \times \nat^\nat \to \nat^\nat$ satisfying \autoref{eq:h} 
can be represented as a  ca-morphism
from $\cds$ to $\ads$ in the following diagram:
\[
  \xymatrix{
    \nat \times \nat^\nat \ar[d]_{\cds} \ar[r]^\descending &\nat^\nat \\
    \Fds(\nat \times \nat^\nat) \ar[r]^-{\Fds(\descending)} &\Fds(\nat^\nat) \ar[u]_{\ads}
  }
\]
where
the functor $\Fds\colon \Set \to \Set$ is defined by $\Fds(X) = X + \nat \times X$,
and
the coalgebra $\cds$ 
and 
the algebra $\ads$ are defined by
\begin{align*}
  \cds(n, \cons{i}{\cons{j}{xs}}) &= \begin{cases}
    \iota_0(n+1,\cons{j}{xs}) &\text{if }i > j \\
    \iota_1(n, 1,\cons{j}{xs}) &\text{if }i\leq j,
  \end{cases} \\
  \ads(\iota_0(xs)) &= xs, \\
  \ads(\iota_1(n, xs)) &=\cons{n}{xs}.
\end{align*}
The commutativity of the diagram above is equivalent to \autoref{eq:h}.

Although such a ca-morphism uniquely exists, this does not follow directly from recursiveness of~$\cds$ nor from corecursiveness of~$\ads$ (cf.~\cite{DBLP:journals/mscs/JeanninKS17}).
We demonstrate how \autoref{prop:rec_solution} explains this unique existence, 
reinterpreting properties (i)--(iii) within our framework.

The first property (i) can be formalized as follows:
$\nat^\nat$ is the limit of the functor $\nat^{(-)} \colon \omega^\op \to \Set$ representing the $\omega^\op$-chain $(\cdots \to \nat^2 \to \nat^1 \to \nat^0)$ in $\Set$, where each morphism is given by the projection mapping of prefixes.

\newcommand{\Gds}{G_\descending}
The second property (ii) ensures that the target algebra $\ads$ satisfies the assumption of \autoref{prop:rec_solution}, that is, $\ads$ is in the image of a right adjoint induced by a step.
Define a functor $\Gds\colon \Psh{\omega} \to \Psh{\omega}$ 
by 
\begin{align*}
\Gds(f)(i) &= f(i) + \nat \times \lim_{i' < i} f (i') \\
&\cong \begin{cases}
f(i) + \nat &\text{ if }i = 0 \\
f(i) + \nat \times f(i-1) &\text{ if }i > 0,
\end{cases}
\end{align*}
 for each $f \in \Psh{\omega}$ and $i \in \omega^\op$.
For the adjunction $\Delta \dashv \lim\colon \Psh{\omega} \to \Set$,
we define a step
$\tau\colon \Delta\Fds \Rightarrow \Gds \Delta$ 
by
$\tau_{X}(i) = \mathrm{id} + \mathrm{id} \times \langle \mathrm{id} \rangle_{i' < i}\colon X+\nat \times X \to X + \nat \times \lim_{i' < i}X$.

The property (ii) can now be rephrased as the first statement of the following lemma.
The third property (iii), which is a kind of B\"uchi condition,
gives rise to the second statement.
\begin{lem}
  The following statements hold.
  \begin{enumerate}
    \item There is a $\Gds$-algebra $a'$ s.t.~$\Alg{\lim}(a') = \ads$.
    \item The coalgebra $\Coalg{\Delta}(\cds)$ is recursive.
  \end{enumerate}
\end{lem}
\begin{proof}
  1)
  Define $a'\colon \Gds(\nat^{(-)}) \to \nat^{(-)}$ by
$a'(i)(\iota_0(t)) = t$, 
$a'(i+1)(\iota_1(n, t)) = n:t$, and $a'(0)(\iota_1(n)) = ()$.
Then it satisfies the required condition.

2)
  The graph of $\Coalg{\Delta}(\cds)$ has nodes of the form $(i, (n, xs))$ where $i \in \omega, n \in \nat, xs \in \nat^\nat$.
  Its edges are given by:
  \begin{itemize}
    \item $(m, (n, \cons{i}{\cons{j}{xs}})) \to (m', (n+1, \cons{j}{xs}))$ iff $i > j$ and $m \geq m'$.
    \item $(m, (n, \cons{i}{\cons{j}{xs}})) \to (m', (1, \cons{j}{xs}))$ iff $i \leq j$ and $m > m'$.
  \end{itemize}
  For each node $(i, (n, xs))$,
  any infinite path from it takes the second case ($i \leq j$ and $m > m'$) infinitely often by (iii).
  However, each time the second case is taken, the first component of the node strictly decreases.
  Therefore, the graph is well-founded, and 
  thus $\Coalg{\Delta}(\cds)$ is recursive by \autoref{thm:wf}.
\end{proof}
\autoref{prop:rec_solution} immediately yields the following.
\begin{prop}
  There is a unique ca-morphism from $\cds$ to $\ads$.
\end{prop}
Informally, we can summarize the situation as the following diagram.
The ca-morphism $\ol{\descending}$ uniquely exists since $\Coalg{\Delta}(\cds)$ is recursive,
and it corresponds to the unique ca-morphism $\descending$ from $\cds$ to $\ads$  by \autoref{lem:coalg_to_alg}.
\[
  \xymatrix@C=0.6pc{
    &\Set \lloop{\Fds} \ar@/^2ex/[rrr]^-{\Delta} &  & &\PshNonSmash{\omega} \rloop{\Gds} \ar@/^2ex/[lll]^-{\lim}_-{\text{\shortstack{$\bot$\\ \, }}} \\
    \cds \ar@{.>}[rr]^-\descending & &\ads = \Alg{\lim}(a') &\Coalg{\Delta}(\cds) \ar@{.>}[rr]^-{\ol{\descending}} & &a'
  }
\]

\section{Steps for Familial Functors} \label{sec:steps}
Our abstract framework for non-wellfounded proof systems is based on 
generalized polynomial functors, called \emph{familial functors}, on presheaf categories.
In this section, we introduce familial functors and investigate steps for them.

We begin by briefly recalling several notions about fibrations.
A \emph{discrete fibration} is a fibration $t\colon \cat{E} \to \cat{B}$ whose fibres $\cat{E}_X$ are discrete categories.
A \emph{two-sided discrete fibration} is a pair of functors $(s\colon \cat{I} \to \cat{A}, q\colon \cat{I} \to \cat{B})$ satisfying the following conditions:
\begin{itemize}
  \item For each $a \to s(i)$ in $\cat{A}$
there is a unique lift $i' \to i$ in $\cat{I}_{q(i)}$.
  \item For each $q(i) \to b$ in $\cat{B}$
there is a unique lift $i \to i'$ in $\cat{I}_{s(i)}$.
\item For each $f\colon i \to j$ in $\cat{I}$,
the domain of the unique lift of $s(f)$ and the codomain of that of $q(f)$ are equal,
and the composite of these lifts is equal to~$f$.
\end{itemize}
\begin{prop}[{\cite[Thm.~3.2.3]{LOREGIAN2020496}}] \label{prop:two_sided_df}
  If $(s, q)$ is a two-sided discrete fibration,
  then~$s$ is a fibration and~$q$ is an opfibration. 
\end{prop}
Here the word ``discrete'' refers to the uniqueness of lifts, not
to the categories \(\cat{A}\), \(\cat{I}\), or \(\cat{B}\) being discrete.  
Under equivalences of categories,
it is known that 
discrete fibrations $\cat{E} \to \cat{B}$ correspond to functors $\cat{B}^\op \to \Set$,
and two-sided discrete fibrations $\cat{A} \leftarrow \cat{I} \rightarrow \cat{B}$ correspond to functors $\cat{A}^\op \times \cat{B} \to \Set$ (cf.~\cite{LOREGIAN2020496}).
Concretely, a profunctor $M\colon\cat{A}^{\op}\times\cat{B}\to\Set$
has a category $\cat{I}$ of elements with objects $(a,b,m)$, where $m\in M(a,b)$.
A morphism $(a,b,m)\to(a',b',m')$ is a pair
$f\colon a\to a'$, $g\colon b\to b'$ satisfying
$M(\id_a,g)(m)=M(f,\id_{b'})(m')$.
The two projections form the corresponding two-sided discrete fibration $\cat{A} \leftarrow \cat{I} \rightarrow \cat{B}$.
An element of $M(a,b)$ can be viewed as a transition from $a$ to $b$.
Such elements, represented by objects of $\cat{I}$, will serve as
edge labels in the labelled graphs introduced in
\autoref{def:labelled_graph}.

\subsection{Familial Functors on Presheaves}
For a functor $f\colon \cat{I} \to \cat{J}$ between small categories,
the functor $f^* \coloneqq (-) \circ f^\op\colon \Psh{\cat{J}} \to \Psh{\cat{I}}$
has both a left and a right adjoint $f_! \dashv f^* \dashv f_*$,
given by left and right Kan extensions along $f^\op$, respectively.
It is a fact that these Kan extensions can be computed as limits and colimits over fibres when~$f$ is a (op)fibration.
\begin{prop}[{cf.~\cite[Prop.~5.7,Cor.~5.8]{perrone2022kan}}] \label{prop:fib_kan_lim}
  Let $f\colon \cat{E} \to \cat{B}$ be a fibration (resp.~opfibration).
  For each $F\in \Psh{\cat{E}}$  and $b \in \cat{B}$, 
  there is an isomorphism
  $f_!(F)(b) \cong \colim_{e \in \cat{E}_b^\op}Fe$  (resp.~$f_*(F)(b) \cong \lim_{e \in \cat{E}_b^\op}Fe$).
\end{prop}

We work with a class of functors on presheaf categories
that generalises the usual notion of polynomial functors.
Since our focus is mainly on coalgebras,
we restrict our attention to endofunctors;
nevertheless, all results in this section can be naturally extended to functors between different presheaf categories.
\begin{defi}[familial functor, cf.~\cite{Weber2007}] \label{def:gpf}
  A \emph{polynomial} is 
  a diagram of small categories $\cat{A} \xleftarrow{s} \cat{I} \xxrightarrow{q} \cat{J} \xxrightarrow{t} \cat{A}$ such that 
  $(s, q)$ forms a two-sided discrete fibration and
 ~$t$ is a discrete fibration.
  The \emph{familial functor} of a polynomial $P = (s, q, t)$ is the functor $\polyfunc{P} \coloneqq t_!q_*s^*\colon \Psh{\cat{A}} \to \Psh{\cat{A}}$. 
\end{defi}

By \autoref{prop:two_sided_df} and \autoref{prop:fib_kan_lim}, 
we obtain
\[
\polyfunc{P}(G)(a) \cong \coprod_{j \in \cat{J}_a^\op} \lim_{i \in \cat{I}_j^\op} G(si)
\quad
\text{ for each } G \in \PshTight{\cat{A}} \text{ and } a \in \cat{A},
\]
where $\cat{J}_a$ is the fibre of~$t$ over~$a$ and $\cat{I}_j$ is the fibre of~$q$ over~$j$.
In particular,
when all categories in the polynomial are discrete,
this expression reduces to
 $\coprod_{j \in \cat{J}_a^\op} \prod_{i \in \cat{I}_j^\op} G(si)$.
 Hence familial functors generalize the usual polynomial functors built from coproducts and products.
 Moreover, in this discrete setting, $\polyfunc{P}$ coincides with the notion of polynomial functors on slice categories~\cite{gambino2013polynomial},
 which is another generalisation of the usual polynomial functors.

\begin{rem}
  The two-sided discrete fibration condition on $(s, q)$ can be removed 
  from \autoref{def:gpf}
  without changing 
  the class of functors obtained (cf.~\autoref{prop:pra_poly}).
  A more general notion obtained by removing the restriction that~$t$ is a discrete fibration
  appears as \emph{generalized polynomial functors} in \cite[Def.~4.1]{Fiore12}.
\end{rem}

We now proceed to study steps for familial functors.
\begin{lem} \label{lem:diag}
  Let $P = (s, q, t)$ and $P' = (s', q', t')$ be polynomials,
  and $u, u_1, u_2$ be functors making the following diagram commute
  and the rightmost square a pullback.
    \begin{equation} \label{eq:poly_morph}
      \xymatrix{
        \cat{A}' \ar[d]^u &\cat{I}'  \ar[l]_{s'}  \ar[r]^-{q'} \ar[d]^{u_1} &\cat{J}' \ar[r]^{t'} \ar[d]^{u_2} \pullbackmark{0, 1}{1, 0} &\cat{A}' \ar[d]^u \\
        \cat{A} &\cat{I}  \ar[l]_{s}  \ar[r]^-{q} &\cat{J} \ar[r]^{t} &\cat{A}
      }
    \end{equation}
  Then there is a canonical natural transformation $\tau\colon u^* \polyfunc{P} \Rightarrow \polyfunc{P'} u^*$.
  \end{lem}
\begin{proof}
  We define $\tau$ as the following natural transformation:
    \begin{equation} \label{eq:poly_morph_nat}
      \xymatrix@C=3em{
        \PshNonSmash{\cat{A}'} \ar[r]^{{s'}^*} &\PshNonSmash{\cat{I}'}  \ar[r]^-{q'_{*}}&\PshNonSmash{\cat{J}'} \ar[r]^{t'_!} &\PshNonSmash{\cat{A}'} \\
        \PshNonSmash{\cat{A}} \ar[u]^{u^*} \ar[r]^{s^*} &\PshNonSmash{\cat{I}} \ar@{=>}[ul]|{\cong} \ar[u]^{u_1^*} \ar[r]^-{q_*} &\PshNonSmash{\cat{J}} \ar[r]^{t_!} \ar@{=>}[ul] \ar[u]^{{u_2}^*} &\PshNonSmash{\cat{A}} \ar[u]^{u^*} \ar@{=>}[ul]|{\cong}
      }
    \end{equation}
  The left-hand natural transformation is an isomorphism because $s \circ u_1 = u \circ s'$.
  The middle natural transformation is the mate, with respect to the adjunctions $q'^* \dashv q'_*$ and $q^* \dashv q_*$, of the natural isomorphism 
  $q'^* \circ u_2^* \cong  u_1^* \circ q^*$ induced by the equality $u_2 \circ q' = q \circ u_1$.

For the right-hand natural isomorphism, 
since $t$ is a discrete fibration and the rightmost square in Diagram~\ref{eq:poly_morph} is a
pullback, we have
for each $X\in\Psh{\cat{J}}$
and $a'\in\cat{A}'$,
\[
\begin{aligned}
(u^*t_!X)(a')
  &\cong\coprod_{j\in\cat{J}_{u(a')}}X(j)\\
  &\cong\coprod_{j'\in\cat{J}'_{a'}}X(u_2(j'))
   \cong(t'_!u_2^*X)(a'),
\end{aligned}
\]
and it is natural in $X$ and $a'$.
\end{proof}
This canonical construction of steps for familial functors motivates the following category of polynomials.
\begin{defi}
 The category $\Poly$  is defined as follows.
 Its objects are polynomials $P = (s, q, t)$.
 A morphism $P' \to P$ is a tuple $(u, u_1, u_2)$ of functors 
 forming a commutative diagram as in Diagram~\ref{eq:poly_morph},
 whose rightmost square is a pullback.
 By abuse of notation, we often write~$u$ for the morphism $(u, u_1, u_2)$.
\end{defi}
For readers familiar with parametric right adjoint functors,
we note that
this category $\Poly$ is closely related to parametric right adjoint functors 
and a class of natural transformations for them.
While this connection helps understanding familial functors,
it is not essential for the remainder of the paper; readers who are not interested in parametric right adjoints may skip the rest of this subsection,
except for \autoref{cor:pra_intersection}, which ensures that $\polyfunc{P}$ satisfies the assumptions of  \autoref{thm:wf}.

Recall that a functor $T\colon \PshTight{\cat{A}} \to \PshTight{\cat{A}}$ is a \emph{parametric right adjoint (or shortly, p.r.a.)}
if 
the induced functor $T_1\colon \PshTight{\cat{A}} \to \PshTight{\cat{A}}/T1$ sending~$X$ to $T(!_X)\colon TX \to T1$ has a left adjoint
where $1$ is a terminal object  in $\PshTight{\cat{A}}$.
Note that~$T$ factors as the composite of $T_1$ with the forgetful functor $\PshTight{\cat{A}}/T1 \to \PshTight{\cat{A}}$.
The following fact discussed in \cite[Remark 2.12]{Weber2007} shows
that the class of familial functors coincides with 
that of p.r.a.~functors.
\begin{prop} \label{prop:pra_poly}
  For each polynomial~$P$, the functor $\polyfunc{P}$ is p.r.a.
  Conversely, every p.r.a.~functor $T\colon  \PshTight{\cat{A}} \to \PshTight{\cat{A}}$ is a familial functor of some polynomial (up to isomorphism).
\end{prop}
  The canonical polynomial $(s, q, t)$ corresponding to a given p.r.a.~$T$ can be constructed as follows.
  Let $\cat{A}/T1$ be the category of elements, obtained as the pullback of $\PshTight{\cat{A}}/T1 \to \PshTight{\cat{A}}$ along the Yoneda embedding.
  Define~$t$ to be the functor $\cat{A}/T1 \to \cat{A}$.
  Consider the functor $\cat{A}/T1 \to \PshTight{\cat{A}}$,
  that corresponds to the right adjoint functor $T_1$ with the equivalence $\Psh{\cat{A}/T1} \simeq \PshTight{\cat{A}}/T1$ (see \autoref{prop:equiv_ra}).
  The two-sided discrete fibration $(s, q)$ is then defined
  as the associated one as in the discussion following \autoref{prop:two_sided_df}.

An important property of p.r.a.~functors is that they preserve connected limits~\cite{Weber2007}.
As a consequence, $\polyfunc{P}$ preserves limits of $\omega^\op$-chains,
and thus the final coalgebra of $\polyfunc{P}$ can be constructed as 
the limit of the final $\omega^\op$-chain (cf.~\cite{Adamek1974}).
Moreover,
since preservation of connected limits implies preservation of wide pullbacks, 
we obtain the following.
\begin{cor} \label{cor:pra_intersection}
 For each polynomial~$P$,
 the functor $\polyfunc{P}$ preserves inverse images and intersections. 
\end{cor}

The next proposition shows that
the class of natural transformation induced by morphisms in $\Poly$ 
coincides with that of natural transformations of the form $\tau\colon u^*F \Rightarrow F' u^*$,
whose component at $1$ is an isomorphism
where $F, F'$ are p.r.a.~functors.
See \autoref{ap:vpoly_morph} for the proof.
\begin{prop} \label{prop:vpoly_morph}
  \begin{enumerate}
    \item 
    For any morphism in $\Poly$, 
    the canonical natural transformation $\tau$ given by \autoref{lem:diag} satisfies that $\tau_1$ is an isomorphism.
  \item
  Let $F\colon \PshTight{\cat{A}} \to \PshTight{\cat{A}}$ and $F'\colon \PshTight{\cat{A}'} \to \PshTight{\cat{A}'}$ be p.r.a.~functors 
  and $u\colon \cat{A}' \to \cat{A}$ be a functor.
  For
  any natural transformation $\tau\colon u^*F \Rightarrow F'u^*$
  such that $\tau_1$ is an isomorphism,
  there exists a morphism in $\Poly$ inducing $\tau$ (up to isomorphism).
  \end{enumerate}
\end{prop}

\subsection{Graphs and Steps for Familial Functors} \label{sec:poly_step}

As we saw in the previous subsection,
morphisms in $\Poly$ induce steps for familial functors,
which
in turn
yield  the following correspondence between ca-morphisms
by \autoref{prop:steps} and~\autoref{lem:coalg_to_alg}.
  \begin{prop} 
    There is a functor $\Poly^\op \to \Cat$
    mapping 
    $P \in \Poly$ to $\Coalg{\polyfunc{P}}$
    and a morphism~$u$ to 
    the functor $\Coalg{u^*}$
    induced by
  the canonical step
  $\tau\colon u^* \polyfunc{P} \Rightarrow \polyfunc{P'} u^*$ (see \autoref{prop:steps}).
  In the same way, there is a functor $\Poly \to \Cat$
    mapping 
    $P \in \Poly$ to $\Alg{\polyfunc{P}}$
    and~$u$ to 
    the functor $\Alg{u_*}$.
  \end{prop}
  \begin{proof}
    This holds because the canonical natural transformations in Diagram~\ref{eq:poly_morph_nat}
    respect composition by the triangle identities.
  \end{proof}
  \begin{cor}\label{cor:alg_coalg_crp}
    Let $u\colon P' \to P$ be a morphism  in $\Poly$.
  For each $\polyfunc{P}$-coalgebra~$c$ and each $\polyfunc{P'}$-algebra~$a$,
  there is 
  a bijective correspondence between ca-morphisms from~$c$ to $\Alg{u_*}(a)$
  and those from $\Coalg{u^*}(c)$ to ~$a$.
  \end{cor}

The restriction to familial functors, rather than arbitrary endofunctors,
offers several advantages for analysing graphs of their coalgebras.
For coalgebras of familial functors, each
edge is generated by a specific premise position.  
We first make this label explicit.
\begin{defi}[labelled graph of a coalgebra of a familial functor]
\label{def:labelled_graph}
Let \(P=(s,q,t)\) be a polynomial and let
\(c\colon C\to \polyfunc{P}(C)\) be a coalgebra.
The \emph{labelled graph} of \(c\) is the directed labelled graph defined as follows.
Its nodes are objects \((x,n) \in \GfibTotal{C} \), which are also written as $n \in C(x)$.
Its edges are given by
\[
  \big(n \in C(x)\big)
  \xrightarrow{i}
  \big(n_i \in C(s(i))\big)
\]
for each $i \in \cat{I}$ with $q(i) = j$, where
  $c_x(n)=(j,\langle n_i\rangle_{i\in \cat{I}_j^{\op}})$.
\end{defi}

\begin{prop} \label{prop:graph_next}
  Let $P = (s, q, t)$ be a polynomial and~$c$ be an $\polyfunc{P}$-coalgebra.
  For each node $(x, n)$ in the graph of~$c$ (cf.~\autoref{def:graph}) and each $y \in \cat{A}^\op$,
  we have
\begin{equation} \label{eq:graph_next}
  \Succ{(x, n)}{c}(y) = \{n_i \mid s(i) = y\},
\end{equation}
where $c_x(n) = (j, \langle n_i \rangle_{i \in \cat{I}_j^\op})$.
Consequently, the underlying graph of the labelled graph of~$c$ is the graph of~\(c\),
and the well-foundedness of the graph of~$c$ is equivalent to that of the labelled graph of~$c$.
\end{prop}
\begin{proof}
  By \autoref{lem:graph},
  we have $(j, \langle n_i \rangle_{i \in \cat{I}_j^\op}) \in \polyfunc{P}(\SuccOmit{(x, n)})(x)$, which means $n_i \in \SuccOmit{(x, n)}(s(i))$ for each $i \in \cat{I}_j^\op$.
  Conversely, define a subpresheaf~$B$ of~$C$
  by $B(y) \coloneqq \{n_i \mid s(i) = y\}$ for each $y \in \cat{A}^\op$. It is indeed a subpresheaf: because $(s, q)$ is a two-sided discrete fibration,
  for each
  $n_{i'} \in B(y')$ and $f\colon y \to y'$ in $\cat{A}$, 
  we have a unique morphism $i \to i'$ in $\cat{I}_j$ above~$f$,
  which satisfies 
  $Cf(n_{i'}) = n_{i} \in B(y)$
  by $\langle n_i \rangle_{i \in \cat{I}_j^\op} \in \lim_{i \in \cat{I}_j^\op}C(s(i))$.
  Therefore, $c_x(n) \in \polyfunc{P}(B)(x)$ implies the inclusion $\SuccOmit{(x, n)} \hookrightarrow B$.
\end{proof}

A morphism \(u\colon P'\to P\) in $\Poly$ provides two things:
the correspondence of ca-morphisms in \autoref{cor:alg_coalg_crp},
and
a way to construct a graph over a given graph along a base-change functor $u^*$,
which we establish in the following proposition.
We exploit both in the next section to define GTC for non-wellfounded proofs and to expand proof systems with ordinal annotations.

\newcommand{\graphhom}[1]{\mathcal{G}(#1)}
\begin{prop} \label{prop:graph_morph}
Let \(u=(u,u_1,u_2)\colon P'\to P\) be a morphism in \(\Poly\), and let
\(c\colon C\to \polyfunc{P}(C)\) be a coalgebra.  Then the assignment
\[
  \big(n \in u^*C(x)\big)\mapsto \big(n \in C(u(x))\big)
\]
extends to a homomorphism $\graphhom{u}$ from the labelled graph of
\(\Coalg{u^*}(c)\) to the labelled graph of~\(c\).  It sends an
edge with label \(i \in \cat{I}'\) to the edge with
label \(u_1(i) \in \cat{I}\).  

Moreover, let $e$
be an edge $\big(n_1 \in C(a_1)\big)\xrightarrow{i}\big(n_2\in C(a_2)\big)$ in the labelled graph of~\(c\),
and
let \(a'_1 \in \cat{A}'_{a_1},a'_2\in\cat{A}'_{a_2}\).
Then there is a labelled edge
  from $(n_1 \in u^*C(a_1'))$ to $(n_2 \in u^*C(a_2'))$
in the labelled graph of \(\Coalg{u^*}(c)\) lying above
the edge $e$ if and only if there exists an object
\(i'\in\cat{I}'\) such that
\begin{equation} \label{eq:graph_lift}
  u_1(i') = i, \quad
  s'(i') = a_2', \quad\text{and}\quad
  t'(q'(i'))=a_1'.
\end{equation}
\end{prop}
\begin{proof}
Let
  $(\node{n}{u^*C(x)}) \xrightarrow{i'} (\node{n'}{u^*C(s'(i'))})$
be an edge in the labelled graph of \(\Coalg{u^*}(c)\).  
Then we have $\pi_0 (c_{u(x)}(n)) = u_2(q'(i'))$.
Write
  $c_{u(x)}(n)
  =
  \bigl(u_2(q'(i')),\langle m_i\rangle_{i\in \cat{I}_{u_2(q'(i'))}^{\op}}\bigr)$.
By definition of the canonical step induced by~\(u\), 
it follows that $n' = m_{u_1(i')}$,
and the labelled graph of~\(c\) has the edge
  \[\big(\node{n}{C(u(x))}\big) \xrightarrow{u_1(i')}
  \big(\node{n'}{C(s(u_1(i')))}\big).\]
Since \(u(s'(i'))=s(u_1(i'))\), this gives the claimed homomorphism.
The moreover part follows from the same description.  
\end{proof}
Thus, once an edge in the labelled graph of \(c\) and objects above its
source and target are fixed, the existence of a lifted edge in the
labelled graph of \(\Coalg{u^*}(c)\) is determined by checking
the existence of a premise position that is compatible with
the chosen source and target objects.  
Motivated by this observation,
for \(i\in\cat{I}\), \(a'_1\in\cat{A}'_{t(q(i))}\), and
\(a'_2\in\cat{A}'_{s(i)}\), we write
\[
  a'_1 \traceto{i} a'_2
\]
if there exists \(i'\in\cat{I}'\) satisfying \autoref{eq:graph_lift}.

\section{Abstract Non-Wellfounded Proof Systems} \label{sec:proof_system}
We finally introduce our framework for non-wellfounded proof systems
on top of the theory of coalgebras and their graphs of familial functors
developed in the previous sections.
In this framework,
derivation trees (or pre-proofs) are modelled as coalgebras,
and
global trace conditions (GTC) are formulated as conditions on the graphs of these coalgebras. 
We establish soundness in this framework and characterise the GTC in terms of recursive coalgebras.

\subsection{Polynomials and Abstract Proof Systems} \label{sec:disc_proof}
We employ familial functors to represent proof systems,
inspired by the work of Fiore \cite{Fiore12} for abstract syntax.
\newcommand{\TS}{\mathrm{TS}}
\begin{defi} 
  An \emph{(abstract) proof system}~$P$ 
  consists of small categories $\Jdg$ and~$R$, functors $\arfunc\colon R \to \Cat$, 
  $\prem\colon \GopfibTotal{\arfunc} \to \Jdg$, and
  $\ccl\colon R \to \Jdg$ such that the diagram
  \[
  \Jdg \xleftarrow{\prem} \GopfibTotal{\arfunc} \xrightarrow{\Gopfib{\arfunc}} R \xrightarrow{\ccl} \Jdg,
  \]
  is
  a polynomial in the sense of \autoref{def:gpf};
  that is, $(\prem, \Gopfib{\arfunc})$ is a two-sided discrete fibration and $\ccl$ is a discrete fibration.
  We sometimes write $P = (\prem, \Gopfib{\arfunc}, \ccl)$ for such a proof system.
  An object of $\Jdg$ is called a \emph{judgement}, and an object of~$R$ is called a \emph{rule}.
  For a rule $r \in R$, the object $\ccl(r)$ is called the \emph{conclusion} of~$r$ and $\prem(r, i)$ is called the \emph{$i$-th premise} of~$r$ for each $i \in \arfunc(r)$.
\end{defi}
  Applying the discussion after \autoref{def:gpf},
  we have
  \[
    \prffunc{P}(G)(s) \cong \coprod_{\substack{r \in R_s}} \lim_{i \in \arfunc(r)^\op}G(\prem(r, i)),
  \text{ for each } G \in \Psh{\Jdg}, s \in \Jdg.
  \]

Intuitively, an abstract proof system represents a proof system whose 
judgements are objects of $\Jdg$, and
whose inference rules are given, for each $r \in R$, by
\begin{center}
  \refstepcounter{equation}\label{eq:rule_r}
  \begin{minipage}{0.95\linewidth}
    \centering
\begin{prooftree}
  \AxiomC{$\{\prem(r, i)\}_{i \in \arfunc(r)}$}
  \RightLabel{$r \in R$.}
  \UnaryInfC{$\ccl(r)$}
\end{prooftree}
\end{minipage}(\theequation)
\end{center}

Morphisms in $\Jdg$ and~$R$ represent admissible transformations between judgements and rules, respectively. This allows us to model proof systems with additional structures, such as parameters or contexts.
Such
situations arise, for example, when interpreting open terms relative to
assignments of free variables; see \autoref{sec:santocanale} for an
example where a non-discrete category of judgements is used. 
We also exploit this feature in \autoref{sec:ord_system} when handling proof systems with ordinal annotations.
For many concrete proof systems, however,
it suffices to take $\Jdg$ and~$R$ to be discrete,
with
an arity function
$\arfunc\colon R \to \nat \, (\hookrightarrow \Cat)$ specifying the number of premises for each rule.
Here $\nat$ is regarded as a discrete category, and its inclusion in
$\Cat$ sends $n$ to the discrete category with $n$ objects
(in particular, $0$ is sent to the empty category).

In the remainder of this section, we fix an abstract proof system $P = (\prem, \Gopfib{\arfunc}, \ccl)$.

\begin{defi}
A \emph{pre-proof} of~$P$
is a coalgebra $c\colon C \to \prffunc{P}(C)$ of the familial functor $\prffunc{P}\colon \Psh{\Jdg} \to \Psh{\Jdg}$.
\end{defi}

The labelled graph of~$c$ is 
a (possibly infinite) derivation graph 
with respect to the inference rules of the form~\autoref{eq:rule_r}.
Concretely, 
a node is $n \in C(s)$, 
representing the pair of a judgement $s \in \Jdg$ and an element $n \in C(s)$.
The structure map~$c$ assigns to each node $n \in C(s)$ a rule $r \in R_s$, i.e.~$\ccl(r) = s$, together with its children (or premise nodes) $\langle \node{n_i}{C(\prem(r, i))}\rangle_{i \in \arfunc(r)^\op}$.
Then for each $i \in \arfunc(r)$, there is an edge 
\[
(n \in C(s)) \xrightarrow{i} (n_i \in C(\prem(r, i))).
\]

When the judgement and rule categories and all arities are discrete,
these data are precisely a derivation graph with a chosen rule at each
node and one outgoing edge for each premise position.
In the general presheaf setting, the assignments of
rules and children must additionally be natural in the judgements.

In categorical semantics, 
a (finite) derivation tree $\Pi$ with a root judgement $\Gamma \vdash \varphi$ is typically interpreted as a morphism $\llbracket \Pi \rrbracket\colon \llbracket \Gamma \rrbracket \to \llbracket \varphi \rrbracket$ in a semantic category $\cat{C}$.
Such an interpretation can be viewed as assigning to each node of the derivation tree a morphism that is compatible with the interpretation of inference rules.
We adopt an analogous approach for abstract proof systems: 
we interpret a pre-proof as a ca-morphism into an algebra that specifies the interpretation of rules.
\begin{defi} \label{def:solution}
  Let $a\colon \polyfunc{P}(\Omega) \to \Omega$ be an algebra.
  A \emph{solution of a pre-proof~$c$ with respect to~$a$}
  is 
  a ca-morphism from~$c$ to~$a$.
\end{defi}
An algebra $a\colon \polyfunc{P}(\Omega) \to \Omega$ specifies a semantics for the proof system~$P$.
The presheaf $\Omega \in \Psh{\Jdg}$ assigns to each judgement the set of its semantic values. 
This is analogous to the hom-set $\cat{C}(\llbracket \Gamma \rrbracket, \llbracket \varphi \rrbracket)$ in ordinary categorical semantics (cf.~\cite{10.5555/373919.373928}), when the judgement~$s$ has the form $\Gamma \vdash \varphi$. 
The structure map~$a$ of the algebra determines, for each rule, how to construct a semantic value for the conclusion node from those for the premises.

Accordingly,
a solution of a pre-proof~$c$ is a consistent assignment of semantic values to the nodes of the derivation graph, respecting the interpretation of rules given by~$a$.

By \autoref{thm:wf}, 
a pre-proof~$c$ is recursive if and only if 
it is well-founded, which is equivalently characterized by the well-foundedness of the graph of~$c$.
In this case,
the universality of the recursive coalgebra ensures the existence of
a unique solution.
\begin{thm} \label{thm:wf_solution}
  If a pre-proof~$c$ is well-founded, then there exists a unique solution of~$c$ w.r.t.~any $\prffunc{P}$-algebra~$a$.
\end{thm}
This result subsumes the standard situation of ordinary finitary proof systems: every finite derivation tree admits a unique solution.

\begin{exa}[descending sequence] \label{eg:desc_proof}
  We revisit the example from \autoref{sec:desc} as a running example. 
  Consider
  the abstract proof system $1 \leftarrow 1+\nat = 1+\nat \to 1$,
  where $1$ is the terminal category and $\nat$ is the discrete category of natural numbers.
This proof system has a single judgement $* \in 1$ and 
the following trivial inference rules indexed by $r \in 1 + \nat$:
\begin{center}
  \AxiomC{$*$}
  \RightLabel{$r$}
  \UnaryInfC{$*$}
  \DisplayProof.
\end{center}
The proof system itself is very simple; the non-trivial information is carried by the coalgebra~$\cds$.
The coalgebra $\cds$ defined in \autoref{sec:desc} is a pre-proof of this proof system,
and the function $\descending$ is a solution of $\cds$ w.r.t.~the algebra $\ads$.
The labelled graph of the pre-proof $\cds$ has
$(*, (n, xs))$ as nodes for each $n \in \nat$ and $xs \in \nat^\omega$.
Its edges are given by
\begin{align*}
\big((n, i:j:t) \in \nat \times \nat^\nat\big) &\xrightarrow{* \in \arfunc(*)} \big((n+1, j:t)  \in \nat \times \nat^\nat\big) & \text{if } i > j, \\
\big((n, i:j:t)  \in \nat \times \nat^\nat\big) &\xrightarrow{* \in \arfunc(n)} \big((1, j:t)\in \nat \times \nat^\nat\big)  &\text{if } i \leq j.
\end{align*}
Thus infinite paths in the graph of \(c_{ds}\) represent infinite call
sequences of the definition of~\(ds\), showing that
\(c_{ds}\) is not well-founded.  
Later, the GTC for this
example will say exactly that along every infinite call sequence, the
output-producing 
case \(i\leq j\) occurs infinitely often.
\end{exa}

\subsection{Global Trace Conditions}
As established in \autoref{thm:wf_solution}, a pre-proof admits a unique solution whenever it is well-founded.
For non-wellfounded pre-proofs, however, the existence of a unique solution does not hold in general.
Nevertheless, by imposing certain conditions on derivation graphs, called \emph{global trace conditions}, we can ensure the unique existence of solutions even for non-wellfounded cases.

To define a global trace condition,
we first introduce a structure for tracing formulas along paths in derivation graphs.
\begin{defi}[trace structure]
  A \emph{trace structure} for~$P$ consists of the following data:
  \begin{itemize}
    \item a presheaf $\fml \in \Psh{\Jdg}$, 
    whose elements $\varphi \in \fml(s)$ are called \emph{formulas},
    \item 
    a proof system~$T$ and 
    a morphism $\Gfib{\fml}\colon T \to P$ in $\Poly$ as below,
    which specifies how to trace formulas along edges in the proof system,
\begin{equation} \label{eq:trace_struct}
  \xymatrix{
    \GfibTotal{\fml} \ar[d]^{\Gfib{\fml}} &\GopfibTotal{\arfunc'}  \ar[l]_{\prem'}  \ar[r]^-{\Gopfib{\arfunc'}} \ar[d]^{(\Gfib{\fml})_1} &R' \ar[d]^{(\Gfib{\fml})_2} \ar[r]^-{\ccl'} \pullbackmark{0, 1}{1, 0} &\GfibTotal{\fml} \ar[d]^{\Gfib{\fml}} \\
    \Jdg &\GopfibTotal{\arfunc}  \ar[l]_{\prem}  \ar[r]^-{\Gopfib{\arfunc}} &R \ar[r]^{\ccl} &\Jdg
  }
\end{equation}
    \item 
    a subset \(\prog \subseteq \coprod_{i \in \GopfibTotal{\arfunc}}(\traceto{i})\), 
    whose elements are said to be \emph{progressing},
  \end{itemize}
satisfying
     for each $h\colon r_1 \to r_2$ in~$R$ and $i\in \arfunc(r_1)$, 
    \begin{center}
      $\big(\varphi \traceto{\arfunc(h)(i)\in \arfunc(r_2)} \psi\big)$ \ (resp.~$\in \prog$)\quad implies \quad
     $\big(\fml(\ccl(h))(\varphi) \traceto{i \in \arfunc(r_1)}\psi\big)$ \ (resp.~$\in \prog$), 
    \end{center}
    and 
for each $r\in R$, $i\in\arfunc(r)$,
and $h\colon s'\to\prem(r,i)$ in $\Jdg$,
    \begin{center}
     $\big(\varphi \traceto{h^*i\in \arfunc(r)} \psi\big)$ \ (resp.~$\in \prog$)\quad implies $\quad \exists \psi' \in \fml(h)^{-1}(\psi).~\big(\varphi \traceto{i \in \arfunc(r)}\psi'\big)$ \ (resp.~$\in \prog$),
    \end{center}
    where $h^*i\to i$ is the unique lift of $h$ in $\arfunc(r)$.
     Please note that $\prem(r_1, i) = \prem(r_2, \arfunc(h)(i))$ since $(\prem, \Gopfib{\arfunc})$ is a two-sided discrete fibration.

  We concisely write $(\Gfib{\fml}, \prog)$ for these data.
    \end{defi}

\begin{defi}[trace]
  Let $(\Gfib{\fml}, \prog)$ be a trace structure and $c$ be a pre-proof of~$P$.
  Consider an infinite path 
in the labelled graph of a pre-proof~$c$
  \[
(n_0 \in C(s_0)) \xrightarrow{j_0 \in \arfunc(r_0)} (n_1 \in C(s_1)) \xrightarrow{j_1 \in \arfunc(r_1)} (n_2 \in C(s_2)) \xrightarrow{j_2 \in \arfunc(r_2)} \cdots.
  \]
  A \emph{trace} along this path is a sequence $\langle \varphi_i \rangle_{i \in \nat}$ of formulas $\varphi_i \in \fml(s_i)$ such that 
  \[
  \varphi_0 \traceto{j_0 \in \arfunc(r_0)} \varphi_{1} \traceto{j_{1} \in \arfunc(r_{1})} \varphi_{2} \traceto{j_{2} \in \arfunc(r_{2})} \cdots.
  \]
  It is called \emph{infinitely progressing} if 
  $\big(\varphi_i \traceto{j_i \in \arfunc(r_i)} \varphi_{i+1}\big) \in \prog$ for infinitely many $i \in \nat$.
\end{defi}
 By \autoref{prop:graph_morph},
a trace along a given infinite path can be equivalently seen as a lift of that path to the labelled graph of $\Coalg{\Gfib{\fml}^*}(c)$ along the graph homomorphism $\graphhom{\Gfib{\fml}}$.
Intuitively, a trace follows a single path in the derivation graph while keeping track of additional information, such as formulas, specified by the trace structure.
Progressing trace steps identify those steps along the trace where a distinguished event occurs, typically corresponding to an unfolding of a fixed point.

\begin{defi}[GTC and proofs] \label{def:GTC}
  A pre-proof~$c$ of~$P$ satisfies the \emph{global trace condition} (or shortly, \emph{GTC})
  on a trace structure 
   if for each infinite path in the labelled graph of~$c$, 
   there exists an infinitely progressing trace along a suffix of the given infinite path.
  A pre-proof is called a \emph{proof} if it satisfies the GTC.
\end{defi}

Well-foundedness trivially implies the GTC.
\begin{prop} \label{prop:wf_GTC}
  Let~$c$ be a pre-proof of~$P$.
  If~$c$ is well-founded, then it satisfies the GTC for each trace structure.
\end{prop}
We later show that the converse also holds; this yields a characterisation of the recursiveness in terms of GTC.

\begin{exa}[GTC for descending sequence] \label{eg:desc_GTC}
  Consider \autoref{eg:desc_proof}.
We define a trace structure as 
$(\id_P, \prog)$
where
$\prog \coloneqq \{(* \traceto{* \in \arfunc(n)} *) \mid n \in \nat\}$.
Although the pre-proof~$\cds$ defined in \autoref{sec:desc} is not recursive,
it satisfies the GTC:
for any infinite path $\langle (n_i, \mathrm{xs}_i) \in \nat \times \nat^\nat \rangle_{i \in \nat}$ with labels $\langle \node{*}{\arfunc(r_i)} \rangle_i$ in the labelled graph of $\cds$,
there is a trivial trace $\langle * \rangle_{i \in \nat}$ along the path.
It is infinitely progressing because 
if
the rule $*$ is applied infinitely many times in a row, 
then the first components of the stream $\mathrm{xs}_i$ form an infinite decreasing sequence of $\nat$, yielding a contradiction.  
\end{exa}

\subsection{Soundness} \label{sec:ord_system}
To establish soundness of proofs in our framework,
we introduce the category $\OJdg{\gamma}$ of judgements with assignments to ordinal numbers,
and an abstract proof system for such judgements.
This construction is
inspired by proof systems with ordinal annotations, which are widely employed in soundness proofs for ordinary non-wellfounded proof systems.
\begin{defi} \label{def:ojdg}
Let $t\coloneqq (\Gfib{\fml}, \prog)$ be
a trace structure for 
$P$ as in Diagram~\ref{eq:trace_struct},
and $\gamma$ be an ordinal number.
The category $\OJdg{\gamma}$ is 
defined as follows:
\begin{itemize}
  \item 
an object is a pair of
an object $s \in \Jdg$ and
a function 
$f\colon \fml(s) \to \gamma \, (= \{\alpha \mid \alpha < \gamma\})$,
\item
a morphism from $(s_1, f_1)$ to $(s_2, f_2)$ 
is a morphism $h\colon s_1 \to s_2$ in $\Jdg$
such that
$f_1\circ \fml(h) \leq f_2$ pointwise, i.e.~$f_1(\fml(h)(\varphi)) \leq f_2(\varphi)$ for each $\varphi \in \fml(s_2)$.
\end{itemize}

\end{defi}
Let $u^\gamma\colon \OJdg{\gamma} \to \Jdg$ be the forgetful functor. 
This forgetful functor $u^\gamma$ is an opfibration because, for each $h\colon s \to s'$ in $\Jdg$ and $f\colon \fml(s) \to \gamma$ in $(\OJdg{\gamma})_{s}$, 
there is an opcartesian morphism over $h$, with codomain $f \circ \fml(h)$.
Therefore,
the functor $(u^\gamma)_*\colon \Psh{\OJdg{\gamma}} \to \Psh{\Jdg}$ can be computed via limits by \autoref{prop:fib_kan_lim}.

We define an abstract proof system $P^\gamma$
as the top row of the following diagram; the functors in the diagram are defined below:
\begin{equation} \label{eq:ojdg_jdg}
  \xymatrix{
    \OJdg{\gamma} \ar[d]^{u^\gamma} &\GopfibTotal{(\arfunc^\gamma)}  \ar[l]_{\prem^\gamma}  \ar[r]^-{\Gopfib{(\arfunc^\gamma)}} \ar[d]^{u^\gamma_1} &R^\gamma \ar[d]^{u^\gamma_2} \ar[r]^-{\ccl^\gamma} \pullbackmark{0, 1}{1, 0} &\OJdg{\gamma} \ar[d]^{u^\gamma} \\
    \Jdg &\GopfibTotal{\arfunc}  \ar[l]_{\prem}  \ar[r]^-{\Gopfib{\arfunc}} &R \ar[r]^{\ccl} &\Jdg
  }
\end{equation}
We define $R^\gamma$
as the category  whose objects are 
pairs of $r \in R$ and $f\colon \fml(\ccl(r)) \to \gamma$ and whose morphisms
$h\colon (r_1, f_1) \to (r_2, f_2)$ are
morphisms $h\colon r_1 \to r_2$ in~$R$ such that 
$\ccl(h)\colon (\ccl(r_1), f_1) \to (\ccl(r_2), f_2)$ in $\OJdg{\gamma}$.
The functors $\ccl^\gamma, u^\gamma_2$ are obvious forgetful functors.

The functor
$\arfunc^\gamma\colon R^\gamma \to \Cat$
is defined
by mapping 
$(r, f) \in R^\gamma$
to
the category $\arfunc^\gamma(r, f)$ defined as follows.
Its objects are pairs of $i \in \arfunc(r)$ and $g\colon \fml(\prem(r, i)) \to \gamma$ such that
for each $\varphi \traceto{i \in \arfunc(r)} \psi$,
\begin{equation} \label{eq:ol_ar}
  \begin{cases}
g(\psi) < f(\varphi) &\text{ if }  \big(\varphi \traceto{i \in \arfunc(r)} \psi\big) \in \prog, \\ 
g(\psi) \leq f(\varphi) &\text{ otherwise}.
  \end{cases}
\end{equation}
Note that the premises of a rule $(r,f) \in R^\gamma$ include every compatible assignment
$g$ for each original premise position $i$. The limit in the familial
functor collects a compatible family of semantic values over all these
annotated premises.
It is suited to ordinal approximation of fixed points: the progressing case
collects the earlier stages needed to interpret an approximation stage
by a limit, as in \autoref{sec:mu_p_gamma}.

A morphism $(i_1, g_1) \to (i_2, g_2)$ in $\arfunc^\gamma(r, f)$ is 
a morphism $h\colon i_1 \to i_2$ in $\arfunc(r)$ such that 
$\prem(r, h)$ is a morphism
$(\prem(r, i_1), g_1) \to (\prem(r, i_2), g_2)$ in $\OJdg{\gamma}$.
For a morphism $h\colon (r_1, f_1) \to (r_2, f_2)$ in $R^\gamma$,
we define
$\arfunc^\gamma(h)\colon \arfunc^\gamma(r_1, f_1) \to \arfunc^\gamma(r_2, f_2)$
by 
\[
\arfunc^\gamma(h)(i, g) \coloneqq (\arfunc(h)(i), g)
\quad \text{ and } \quad 
\arfunc^\gamma(h)(f) \coloneqq \arfunc(h)(f),
\]
for each $f\colon (i, g) \to (i', g')$ in $\arfunc^\gamma(r_1, f_1)$.
Note that $(\arfunc(h)(i), g) \in \arfunc^\gamma(r_2, f_2)$ holds
because 
for each $\varphi \traceto{\arfunc(h)(i) \in \arfunc(r_2)} \psi$,
it follows that
$g(\psi) \leq f_1\big(\fml(\ccl(h))(\varphi)\big) \leq f_2(\varphi)$
by assumption on the trace structure,
and the first inequality is strict if $(\varphi \traceto{\arfunc(h)(i) \in \arfunc(r_2)} \psi) \in \prog$.

The functor $u^\gamma_1$ is the canonical forgetful functor,
induced by the natural transformation $\pi_0\colon \arfunc^\gamma \Rightarrow \arfunc \circ u^\gamma_2$.
The functor $\prem^\gamma$ is defined by
$\prem^\gamma((i, g) \in \arfunc^\gamma(r, f)) \coloneqq (\prem(r, i), g)$.
Using these data,
$u^\gamma\colon P^\gamma \to P$ forms a morphism in $\Poly$.

By \autoref{cor:alg_coalg_crp},
this morphism
induces functors $\Alg{u^\gamma_*}$ and $\Coalg{(u^\gamma)^*}$,
together with a canonical bijective correspondence between ca-morphisms~$h$ and $h^\gamma$ as illustrated below for each $\polyfunc{P^\gamma}$-algebra $a'$.
\[
  \xymatrix@C=0.6pc{
    &\Psh{\Jdg} \lloop{\polyfunc{P}} \ar@/^2ex/[rrr]^-{(u^\gamma)^*} & & &\Psh{\OJdg{\gamma}} \rloop{\polyfunc{P^\gamma}} \ar@/^2ex/[lll]^-{u^\gamma_*}_-{\text{\shortstack{$\bot$\\ \, }}} \\
    c \ar[rr]^-h & &\Alg{u^\gamma_*}(a') &\Coalg{(u^\gamma)^*}(c) \ar[rr]^-{h^\gamma} & &a'
  }
\]

The pre-proof $\Coalg{(u^\gamma)^*}(c)$ contains all ordinal annotations
of nodes and all compatible annotated premise edges.
This construction of $P^\gamma$ is complementary to the trace
structure at the level of infinite paths.
For an infinite path $\pi$ in the labelled graph of a pre-proof $c$,
a suffix of $\pi$ admits an infinitely progressing trace
if and only if, for every ordinal $\gamma$,
the infinite path $\pi$ has no infinite lift to the labelled graph of
$\Coalg{(u^\gamma)^*}(c)$.
The two directions follow from the arguments in the proofs of
\autoref{lem:GTC_rec} and \autoref{lem:wf_rec}.
We exploit this property to establish soundness
(\autoref{thm:soundness}) and characterise the GTC
in terms of recursiveness (\autoref{thm:GTC_rec}).

\begin{lem} \label{lem:GTC_rec}
  If a pre-proof~$c$ satisfies the GTC on a trace structure~$t$, then
  $\Coalg{(u^\gamma)^*}(c)$ is recursive.
\end{lem}
\begin{proof}
  We prove the statement by contraposition. 
  Assume that there exists an infinite path 
  $\langle \node{n_i}{((u^\gamma)^*C)(s_i, f_i)} \rangle_{i \in \nat}$  with labels $\langle \node{(j_i, g_i)}{\arfunc^\gamma(r_i, f_i)} \rangle_{i \in \nat}$
  in the labelled graph of $\Coalg{(u^\gamma)^*}(c)$.
  Then 
  $\langle \node{n_i}{C(s_i)} \rangle_{i \in \nat}$ with $\langle \node{j_i}{\arfunc(r_i)} \rangle_{i \in \nat}$
  is also an infinite path in the labelled graph of~$c$.

  If~$c$ satisfies the GTC,
  there exists $N \in \nat$ and
  an infinitely progressing trace 
  $\langle \varphi_i \rangle_{i \geq N}$ along the tail of the infinite path from~$N$.
  By definition of $P^\gamma$,
  for each $i \geq N$,
  $f_{i}(\varphi_i) \geq g_i(\varphi_{i+1}) = f_{i+1}(\varphi_{i+1})$ 
  since $\varphi_i \traceto{j_i \in \arfunc(r_i)} \varphi_{i+1}$,
  and the inequality is strict if $(\varphi_i \traceto{j_i \in \arfunc(r_i)} \varphi_{i+1}) \in \prog$.
  Since the trace is infinitely progressing,
  we obtain a strictly decreasing sequence of ordinals,
  which is a contradiction.
\end{proof}

By combining this lemma with \autoref{lem:coalg_to_alg}, we obtain the soundness theorem for non-wellfounded proofs.
\begin{thm}[Soundness] \label{thm:soundness}
    If~$c$ is a proof of~$P$, then~$c$ has a unique solution with respect to $\Alg{u^\gamma_*}(a')$ for any $\polyfunc{P^\gamma}$-algebra $a'$.
\end{thm}
As discussed after \autoref{def:solution}, 
an $\polyfunc{P}$-algebra in the definition of a solution
specifies a semantics for the proof system.
In typical examples of non-wellfounded proofs (as in \autoref{sec:modal}),
the existence of $a'$ comes from the fact that the semantics of a fixed-point operator $\mu f$ is given by $f^\gamma(\bot)$, which is the limit of transfinite iteration $f^\alpha(\bot)$ for ordinals~$\alpha$ less than or equal to $\gamma$.

\subsection{Coalgebraic Characterisation of GTC}
In this section, we investigate the relationship between recursiveness and GTC,
providing two characterisations:
recursiveness in terms of GTC,
and GTC in terms of recursiveness.

We begin with the former: the characterisation of recursiveness in terms of GTC.
The following result shows that the implication in \autoref{prop:wf_GTC} is in fact an equivalence.
\begin{prop} \label{prop:rec_GTC}
  A pre-proof~$c$ is recursive if and only if it satisfies the GTC for all trace structures.
\end{prop}
\begin{proof}
  The only-if part is given by \autoref{prop:wf_GTC} and \autoref{thm:wf}.
  For the if-part, consider the trace structure $(\id, \emptyset)$ where no edges are progressing.
  If there is an infinite path in the labelled graph of~$c$,
  then there is no infinitely progressing trace along the path,
  contradicting the assumption that~$c$ satisfies the GTC on this trace structure.
\end{proof}

We next turn to the converse direction, giving a characterisation of GTC in terms of recursiveness.
\newcommand{\height}{\mathrm{ht}_{\vec{j}}}
\begin{defi} \label{def:non_progressing}
  Let~$c$ be a pre-proof of~$P$.
  An infinite path in the labelled graph of $c$
  is \emph{non-progressing}
  if no suffix of it admits an infinitely progressing trace.
  For such an infinite path with labels $\vec{j} = \langle \node{j_i}{\arfunc(r_i)} \rangle_{i \in \nat}$,
  we define a directed graph $G_{\vec{j}}$ whose nodes are pairs $(i, \varphi)$ of $i \in \nat$ and $\varphi \in \fml(\ccl(r_i))$,
  and whose edges are given by $(i, \varphi_i) \to (k, \varphi_k)$ 
  when $i < k$ and there is $\varphi_i \traceto{j_i \in \arfunc(r_i)} \cdots \traceto{j_{k-1} \in \arfunc(r_{k-1})} \varphi_k$ 
  such that only the last edge is progressing.

The graph $G_{\vec{j}}$ is well-founded because otherwise, there exists an infinitely progressing trace, contradicting our assumption.
  For each node $(i, \varphi)$,
  the \emph{height $\height(i, \varphi)$} of the node is an ordinal number defined by transfinite induction as follows:
  \[
    \height(i, \varphi) \coloneqq \sup\{\height(j, \psi) + 1 \mid (i, \varphi) \to (j, \psi) \text{ in }G_{\vec{j}}\},
  \]
  with the convention that $\sup \emptyset = 0$.
\end{defi}
  It satisfies 
  $\height(i, \varphi) > \height(j, \psi)$ 
  for each edge $(i, \varphi) \to (j, \psi)$ in $G_{\vec{j}}$
  and 
  $\height(i, \varphi) \geq \height(i+1, \psi)$ 
  for each $\varphi \traceto{j_i \in \arfunc(r_i)} \psi$.
\begin{lem} \label{lem:wf_rec}
  Let~$c$ be a pre-proof of~$P$.
    Assume that $\gamma$ is an ordinal 
    greater than
    the height of any node for each non-progressing
    infinite path.
    If $\Coalg{(u^\gamma)^*}(c)$ is recursive,
    then~$c$ satisfies the GTC.
\end{lem}
\begin{proof}
  Assume that~$c$ does not satisfy GTC.
  Then there is an infinite path $\langle \node{n_i}{C(s_i)} \rangle_{i \in \nat}$ with labels $\vec{j} \coloneqq \langle \node{j_i}{\arfunc(r_i)} \rangle_{i \in \nat}$ in the labelled graph of~$c$
  such that
  no suffix of this path admits an infinitely progressing trace.
  For each $i \in \nat$,
  define $f_i\colon \fml(s_i) \to \gamma$ by
  $f_i(\varphi) \coloneqq \height(i, \varphi)$.
  Then 
  it follows that
  $\langle \node{n_i}{((u^\gamma)^*C)(s_i, f_i)} \rangle_{i \in \nat}$
  is an infinite path in the graph of $\Coalg{(u^\gamma)^*}(c)$.
\end{proof}

\begin{thm} \label{thm:GTC_rec}
  A pre-proof~$c$ satisfies the GTC if and only if $\Coalg{(u^\gamma)^*}(c)$ is recursive for each ordinal number $\gamma$.
\end{thm}

\begin{exa}
  Consider the running example of descending sequences (\autoref{eg:desc_proof} and \autoref{eg:desc_GTC}).
  In this case, 
  an infinite path in the graph of a pre-proof~$c$ has labels 
  $\vec{j}$ of the form $\langle \node{*}{\arfunc(r_i)} \rangle_{i \in \nat}$
  with $r_i \in 1 + \nat$.
  A sequence $\vec{j}$ is non-progressing if there is $N \in \nat$ such that $r_i = *$  for each $i \geq N$.
  For any such sequence,
  the height is at most~$N$.
  Hence taking $\gamma \coloneqq \omega$ satisfies the assumption of \autoref{lem:wf_rec},
  and we have that
  the pre-proof $\cds$ satisfies the GTC if and only if $\Coalg{(u^\omega)^*}(\cds)$ is recursive.
  Moreover, 
  in this situation
  the construction of the proof system $P^\omega$ 
  offers
  the adjunction $\Delta \dashv \lim$ and the endofunctor $\Gds$ described in \autoref{sec:desc}.
\end{exa}

\section{Base Change of GTC} \label{sec:base_change}
As a supplementary property of the GTC,
let us introduce a way of transferring trace structures and GTCs along morphisms between proof systems.
This is analogous to the classical result (\autoref{prop:rec_coalg}) on the transfer of recursiveness of coalgebras along adjunctions.
\begin{lem} \label{lem:poly_pb}
  The category $\Poly$ has pullbacks.
\end{lem}
\begin{proof}
  Let $(f_0, f_1, f_2)\colon P_2 \to P_1$ and $(g_0, g_1, g_2)\colon P_3 \to P_1$ 
  be morphisms
  in $\Poly$. 
  We write 
  $\cat{C}^0_i \xleftarrow{s_i} \cat{C}^1_i \xrightarrow{q_i} \cat{C}^2_i \xrightarrow{t_i} \cat{C}^0_i$ for the polynomials $P_i$ ($i=1, 2, 3$).
  For $j=0, 1, 2$, define a category $\cat{C}^j$ by the change-of-base of $g_j$ along $f_j$,
  and write $h^j_i\colon \cat{C}^j \to \cat{C}^j_i$ ($i = 2, 3$) for the projection functors.
  We define $P = \cat{C}^0 \xleftarrow{s} \cat{C}^1 \xrightarrow{q} \cat{C}^2 \xrightarrow{t} \cat{C}^0$ as the induced functors 
  by universality of pullbacks.
  Then the pullback lemma shows that $(h^2_i, t, t_i, h^0_i)$ forms a pullback square for each $i = 2, 3$.
  Since discrete fibrations are stable under pullbacks,~$t$ is a discrete fibration.
  By definition of $(s, q)$, it forms a two-sided discrete fibration.
  Thus,~$P$ is a polynomial and $(h^0_i, h^1_i, h^2_i)\colon P \to P_i$ are morphisms in $\Poly$ for $i=2, 3$.
  It is straightforward to check that they form a pullback in $\Poly$.
\end{proof}
\begin{defi} \label{def:pullback_ts}
  Let $P_i = \Jdg_i \xleftarrow{\prem_i} \GopfibTotal{\arfunc_i} \xrightarrow{\Gopfib{(\arfunc_i)}} R_i \xrightarrow{\ccl_i} \Jdg_i$ ($i = 1, 2$) be proof systems and
  $f\colon P_2 \to P_1$ be a morphism in $\Poly$.
  Given a trace structure
  $(\Gfib{\fml}\colon T_1 \to P_1, \prog)$  for $P_1$,
  we define a trace structure $(\Gfib{(\fml \circ f^\op)}, \prog_2)$ for $P_2$, denoted $f^\sstar(\Gfib{\fml}, \prog)$,
  as follows.
  \begin{itemize}
    \item The morphism $\Gfib{(\fml \circ f^\op)}\colon T_2 \to P_2$ is given by the pullback (change-of-base) of~$\Gfib{\fml}$ along~$f$. 
    Note that $\fml \in \Psh{\Jdg_1}$ is mapped to $\fml \circ f^\op \in \Psh{\Jdg_2}$ via the pullback.
  \begin{equation} \label{eq:pb_uf}
    \xymatrix{
      T_2 \pullbackmark{0, 1}{1, 0}\ar[d]_{\Gfib{(\fml \circ f^\op)}} \ar[r]^{f'} &T_1 \ar[d]^{\Gfib{\fml}} \\
      P_2 \ar[r]^{f} &P_1.
    }
  \end{equation}
  \item The subset $\prog_2$ is defined by 
  \[
  \prog_2 \coloneqq \{\varphi \traceto{i \in \arfunc_2(r)} \psi \text{ for }T_2, P_2 \mid \varphi \traceto{f_1(i) \in \arfunc_1(f_2(r))} \psi \in \prog \text{ for }T_1, P_1 \}.\]
  \end{itemize}
  The tuple $(\Gfib{(\fml \circ f^\op)}, \prog_2)$ forms a trace structure
  because $f_1$ preserves opcartesian and cartesian morphisms by the uniqueness of opcartesian and cartesian lifts.
\end{defi}

\begin{prop} \label{prop:pullback_GTC}
  Let $f\colon P_2 \to P_1$ be a morphism in $\Poly$,
  and~$c\colon C \to \polyfunc{P_1}(C)$ be a pre-proof of $P_1$.
  Then the following statements hold.
  \begin{enumerate}
    \item If~$c$ is recursive, then $\Coalg{f^*}(c)$ is recursive.
    \item 
      If~$c$ satisfies the GTC 
      for $P_1$ with a trace structure~$t$, then 
      $\Coalg{f^*}(c)$ satisfies the GTC for $P_2$ with $f^\sstar t$.
  \end{enumerate}
\end{prop}
\begin{proof}
  (1) Follows immediately from \autoref{prop:rec_coalg}.

  (2)
  For each infinite path $\langle n_i \in (f^*C)(s_i) \rangle_{i \in \nat}$ with labels $\langle j_i \in \arfunc_2(r_i) \rangle_{i \in \nat}$ in the labelled graph of $\Coalg{f^*}(c)$,
  by \autoref{prop:graph_morph}, we have a path $\langle n_i \in C(f(s_i)) \rangle_{i \in \nat}$ with labels $\langle f_1(j_i) \in \arfunc_1(f_2(r_i)) \rangle_{i \in \nat}$ in the labelled graph of~$c$.
  When~$c$ satisfies the GTC for $t = (\Gfib{\fml}, \prog)$,
  there exists $N \in \nat$ and an infinitely progressing trace $\langle \varphi_i \rangle_{i \geq N}$ along the suffix $\langle n_i \in C(f(s_i)) \rangle_{i \geq N}$
  on~$t$.
  Let us show that
  $\langle (s_i, \varphi_i) \rangle_{i \geq N}$ is an infinitely progressing trace for $\langle n_i \in (f^*C)(s_i) \rangle_{i \geq N}$ on $f^\sstar t$.
  For each $i \geq N$,
  $\varphi_i \traceto{f_1(j_i) \in \arfunc_1(f_2(r_i))} \varphi_{i+1}$ for $P_1$ with~$t$ holds and it means that 
  there exists $\node{i'}{\arfunc_1'(f_2(r_i), \varphi_i)}$ such that 
  $(\Gfib{\fml})_1(\node{i'}{\arfunc_1'(f_2(r_i), \varphi_i)}) = (\node{f_1(j_i)}{\arfunc_1(f_2(r_i))})$ and
  $\prem_1'(\node{i'}{\arfunc_1'(f_2(r_i), \varphi_i)}) = \varphi_{i+1}$.
  By definition of $f^\sstar t$,
  it follows that $(s_i, \varphi_i) \traceto{j_i \in \arfunc_2(r_i)} (s_{i+1}, \varphi_{i+1})$ for $P_2$ with $f^\sstar t$.
  Moreover, this trace is infinitely progressing.
\end{proof}

\section{Examples} \label{sec:example}
The examples in this section are intended to clarify 
how the abstract framework is instantiated in concrete proof systems.
The first 
example, in  \autoref{sec:modal}, shows how the GTC is guided by the semantics in fixed-point
logics.  The ordinal-indexed approximants used to interpret fixed-point
operators provide the lifted algebra required to apply the soundness theorem (\autoref{thm:soundness}).
The second example, in \autoref{sec:santocanale}, illustrates the use of 
non-discrete categories of judgements.
When open terms
are interpreted relative to assignments of free variables, judgements
vary functorially with these assignments, and hence are naturally
organised into a category rather than a mere set.  

\subsection{For Modal $\mu$-calculus} \label{sec:modal}
\newcommand{\Prop}{\mathrm{Prop}}
\newcommand{\Var}{\mathrm{Var}}
We consider a non-wellfounded proof system for the modal $\mu$-calculus drawn from~\cite{DBLP:conf/lics/AfshariL17,AfshariW22}. This system is an adaptation of the tableaux proof system given by Niwi\'nski and Walukiewicz~\cite{DBLP:journals/tcs/NiwinskiW96}.

\subsubsection{sequent calculus for the modal $\mu$-calculus}
We begin by recalling the modal $\mu$-calculus.
Let $\Prop$ be a set of propositional letters,
$\Sigma$ be a set of actions,
and
$\Var$ be a countably infinite set of variables.
The set of $\mu$-calculus formulas is defined by the following BNF:
\[
  \varphi ::= p \mid \neg p \mid x \mid \varphi \lor \varphi \mid \varphi \land \varphi \mid \langle a \rangle \varphi \mid [a] \varphi \mid \mu x.\varphi \mid \nu x.\varphi,
\]
where $p \in \Prop$, $x \in \Var$, and $a \in \Sigma$.
We refer to $(\neg, \lor, \land)$
as logical connectives, to $(\langle a \rangle, [a])$ as modalities, and to $(\mu x, \nu x)$ as fixed-point operators.
A formula $\varphi$ is said to be
\emph{well-named}
if 
 for each variable~$x$, there is at most one subformula of the form $\sigma x.\psi$ with $\sigma \in \{\mu,\nu\}$, and if~$x$ occurs free in $\varphi$, then there is no subformula of $\varphi$ of the form $\sigma x.\psi$.
We regard formulas up to \(\alpha\)-equivalence and always choose
well-named representatives.

Substitutions are defined in the standard way,
renaming bound variables when necessary so that formulas remain well-named.
Given a formula $\varphi$,
we define a partial order $\leq_\varphi$ on variables occurring in $\varphi$
as the least one such that
$x \leq_\varphi y$
if~$x$ is free in a subformula of the form $\sigma y.\psi$ of $\varphi$, where $\sigma \in \{\mu, \nu\}$.
We call~$x$ a \emph{$\sigma$-variable of $\varphi$} if $\sigma x.\psi$ is a subformula of $\varphi$  for some $\psi$.

The negation operator $\neg$ can be extended to all closed formulas by using De Morgan duality, e.g.~$\neg(\mu x.\langle a \rangle x \land p) = \nu x.[a]x \lor \neg p$.
To define the semantics of formulas, we fix a labelled transition system (LTS for short)
 $\mathcal{K} = (S, R\colon \Sigma \to \mathcal{P}(S \times S), L\colon \Prop \to \mathcal{P}S)$.
Let $\rho\colon \Var \to \mathcal{P}S$ be a function, called a \emph{valuation}. 
The \emph{semantics  of a formula $\varphi$ in~$K$ under $\rho$} is given by a set of states $\llbracket \varphi \rrbracket_\rho^\mathcal{K} \subseteq S$ defined inductively as follows:
\begin{align*}
  &\llbracket p \rrbracket_\rho^\mathcal{K} \coloneqq L(p),
  \llbracket x \rrbracket_\rho^\mathcal{K} \coloneqq \rho(x), 
  \llbracket \varphi_1 \land \varphi_2 \rrbracket_\rho^\mathcal{K} \coloneqq \llbracket \varphi_1 \rrbracket_\rho^\mathcal{K} \cap \llbracket \varphi_2 \rrbracket_\rho^\mathcal{K}, \\
  &\llbracket [a] \varphi \rrbracket_\rho^\mathcal{K} \coloneqq \{s \mid \forall t.~\big((s, t) \in R(a) \text{ implies } t \in \llbracket\varphi \rrbracket_\rho^\mathcal{K}\big).\}, \\
  &\llbracket \nu x.\varphi \rrbracket_\rho^\mathcal{K} \coloneqq 
  \mathrm{gfp}(\lambda v \in \mathcal{P}S.~\llbracket \varphi \rrbracket_{\rho[x \mapsto v]}^\mathcal{K}).
\end{align*}
The semantics of other formulas are defined by duality.
Here, for a monotone endofunction~$f$  on a complete lattice,
$\mathrm{gfp}(f)$ denotes its greatest fixed point, whose existence is guaranteed by the Knaster--Tarski fixed-point theorem.
Note that for each formula $\sigma x.\varphi$ with $\sigma \in \{\mu, \nu\}$,
variables occurring in $\varphi$ appear only positively in $\varphi$.
This ensures that
the function $\lambda v \in \mathcal{P}S.~\llbracket \varphi \rrbracket_{\rho[x \mapsto v]}^\mathcal{K}$ is monotone on the complete lattice $(\mathcal{P}S, \subseteq)$.

We next recall a sequent calculus introduced in~\cite{DBLP:conf/lics/AfshariL17}.
A \emph{sequent} is a finite set of formulas.
We say that a sequent $\Gamma$ is \emph{valid} in $\mathcal{K}$ if 
$\bigcup_{\varphi \in \Gamma}\llbracket \varphi \rrbracket_\rho^\mathcal{K} = S$ for each valuation $\rho$.
The proof system consists of the following inference rules:
\begin{center}
  \AxiomC{}
  \RightLabel{Ax}
 \UnaryInfC{$p, \neg p$}
 \DisplayProof
 \quad
  \AxiomC{$\Gamma$}
  \RightLabel{Wk}
 \UnaryInfC{$\Gamma, \varphi$}
 \DisplayProof
 \quad
 \AxiomC{$\Gamma, \varphi, \psi$}
 \RightLabel{$\lor$}
 \UnaryInfC{$\Gamma, \varphi \lor \psi$}
 \DisplayProof
 \quad
 \AxiomC{$\Gamma, \varphi$}
 \AxiomC{$\Gamma, \psi$}
 \RightLabel{$\land$}
 \BinaryInfC{$\Gamma, \varphi \land \psi$}
 \DisplayProof
 \\
 \AxiomC{$\Gamma, \varphi$}
 \RightLabel{Mod}
 \UnaryInfC{$\langle a \rangle \Gamma,[a]\varphi$}
 \DisplayProof
 \quad
 \AxiomC{$\Gamma, \varphi[\sigma x.\varphi/x]$}
 \RightLabel{$\sigma$ (where $\sigma \in \{\mu, \nu\}$)}
 \UnaryInfC{$\Gamma, \sigma x.\varphi$}
 \DisplayProof
\end{center} 
Here, 
\(\langle a\rangle\Gamma\) is the set
\(\{\langle a\rangle\varphi \mid \varphi\in\Gamma\}\). 
The symbols Ax, Wk,  $\lor$, $\land$, Mod, $\sigma$ are rule schemata,
each representing a (possibly infinite) family of concrete rule instances.
A rule instance is obtained by instantiating a rule schema with an assignment to its metavariables.
For example, $(\lor, (\Gamma, \varphi, \psi \coloneqq \emptyset, p, p))$
represents the rule instance whose conclusion is $p \lor p$ and whose premise is~$p$.
For the rule schema $\sigma$,
the bound variable~$x$ in the principal formula $\sigma x.\varphi$ is said to be \emph{unfolded}.

\subsubsection{abstract proof system~$P$} \label{subsec:modal_p}
We then define an abstract proof system~$P$ as follows.
Regarding sets as discrete categories,
we take $\Jdg$ to be the set of sequents and
$R$ to be the set of rule instances.
The arity functor $\arfunc$ is given as $\arfunc\colon R \to \nat \hookrightarrow \Cat$,
where
$\arfunc(r)$ is the number of the premises of the rule~$r$.
For $i \in \arfunc(r)$,
$\prem(r, i)$ and $\ccl(r)$ are defined as the~$i$-th premise and the conclusion of the rule instance~$r$, respectively.

Then 
a (possibly infinite) tree $\Pi$ obtained by applying rules, can be represented as a pre-proof $c\colon C \to \prffunc{P}(C)$ in $\Psh{\Jdg}$
defined by
$C(\Gamma)$ to be the set of vertices of $\Pi$ labelled with the sequent $\Gamma$, and
$c_\Gamma(v) \coloneqq (r, \langle v_i \rangle_{i \in \arfunc(r)})$
where
$v_i$ is the~$i$-th premise node of~$v$ in $\Pi$ and 
$r$ is the rule instance from~$v$ to $\langle v_i \rangle_i$.
A finite derivation graph representing a cyclic pre-proof is encoded in
the same way.
Back-links are then represented by edges to existing vertices.
Unfolding such a coalgebra at a chosen root gives the corresponding
possibly infinite derivation tree.

We move on to define a trace structure for~$P$.
A \emph{marked formula} $\check{\varphi}$ (resp.~\emph{marked sequent}) is a formula (resp.~sequent) equipped with one mark on an occurrence of a fixed-point operator $\nu$ in $\varphi$. We sometimes explicitly write the mark as $\nu_\bullet$.
A trace structure $(\Gfib{\fml}\colon T \to P, \prog)$ for this proof system is defined by
\begin{itemize}
  \item $\fml \in \Psh{\Jdg}$ maps $\Gamma$ to 
  the set of marked sequents of $\Gamma$.
  \item
  A proof system $T = (\GfibTotal{\fml} \xleftarrow{\prem'} \GopfibTotal{\arfunc'}   \xrightarrow{\Gopfib{\arfunc'}} R' \xrightarrow{\ccl'} \GfibTotal{\fml})$
  is defined as a refinement
  of the original proof system~$P$ by explicitly tracking marked formulas:
  $R'$ is defined to be the change of base $R \times_{\Jdg} \GfibTotal{\fml}$,
  and the premises of a rule $(r, \check{\varphi}) \in R'$ 
  are defined by marking the formulas of the premises of~$r$ in~$P$ that correspond to $\check{\varphi}$ in the inclusion.
  For example, 
  \begin{center} \label{eq:modal_marked_rule}
 \AxiomC{$\Gamma, \check{\varphi}, \psi$}
 \RightLabel{$\lor$}
 \UnaryInfC{$\Gamma, \check{\varphi} \lor \psi$}
 \DisplayProof \quad 
 \AxiomC{$\{\Gamma, \varphi[\nu_\bullet x.\varphi/x]_i\}_{i}$}
 \RightLabel{$\nu$}
 \UnaryInfC{$\Gamma, \nu_\bullet x.\varphi$}
 \DisplayProof \quad
 \AxiomC{$\Gamma, \check{\varphi}[\sigma x.\varphi/x]$}
 \AxiomC{$\{\Gamma, \varphi[\sigma x.\check{\varphi}/x]_i\}_i$}
 \RightLabel{$\sigma$}
 \BinaryInfC{$\Gamma, \sigma x.\check{\varphi}$}
 \DisplayProof \quad
  \end{center}
  where the index~$i$ in the last one ranges over the occurrences of~$x$ in $\varphi$,
  and for each such~$i$, 
  the marked formulas
  $\varphi[\nu_\bullet x.\varphi/x]_i$ and $\varphi[\sigma x.\check{\varphi}/x]_i$
  are obtained by marking only  
  in the substituted part for 
  the occurrence~$i$.
  \item A subset $\prog$ is defined by marking the unfolding of $\nu$-variables:
  \[\prog \coloneqq \{\nu_\bullet x.\varphi \traceto{* \in \arfunc(r)} \varphi[\nu_\bullet x.\varphi/x]_i \mid 
  r \text{ is a rule instance of $\nu$ and }
  i \text{ is an occurrence of~$x$ in $\varphi$}\}.\]
\end{itemize}

With this trace structure,
the resulting GTC (see \autoref{def:GTC})
is equivalent to the condition adopted in~\cite{DBLP:conf/lics/AfshariL17}, namely:
every infinite path $\langle \Gamma_i \in \Jdg\rangle_{i \in \nat}$ in $\Pi$ has a $\nu$-thread,
where a $\nu$-thread is defined as follows.
A \emph{thread} along an infinite path $\langle\Gamma_i\rangle_{i \in \nat}$ is a sequence of formulas $\langle\varphi_i \in \Gamma_i \rangle_{i \geq N}$ with $N \in \nat$
such that $\varphi_{i+1}$ follows $\varphi_i$ through the rule
from $\Gamma_i$ to $\Gamma_{i+1}$ for each $i \geq N$.
For example, $\varphi\lor\psi$ is followed by either $\varphi$ or $\psi$
under the $\lor$-rule, and $\sigma x.\varphi$ by
$\varphi[\sigma x.\varphi/x]$ under the $\sigma$-rule;
formulas in an unchanged context are followed by themselves.
A thread is called a \emph{$\nu$-thread} if 
a $\nu$-variable~$x$ is the minimal variable
(with respect to $\leq_{\varphi_N}$)
among variables
unfolded infinitely often along the thread.

We use the following property of traces in the modal
\(\mu\)-calculus, following \cite[Lemma~7]{DBLP:conf/mfcs/Bruse14}: if a thread along an infinite path contains infinitely many unfoldings of a fixed-point operator, 
then there is a unique variable that is minimal with respect to the dependency order among the variables unfolded infinitely often along the thread.

\begin{prop} \label{prop:nu-thread}
  The GTC of the proof system~$P$ equipped with the trace structure $(\Gfib{\fml}\colon T \to P, \prog)$ defined above 
  is equivalent to the following condition:
  \begin{center}
  every infinite path $\langle \Gamma_i \rangle_{i \in \nat}$ in $\Pi$ has a $\nu$-thread.
  \end{center}

\end{prop}
\begin{proof}
 Assume that a pre-proof satisfies the GTC.
 Then for every infinite path $\langle \node{v_i}{C(\Gamma_i)} \rangle_{i \in \nat}$
 with labels $\langle \node{j_i}{\arfunc(r_i)} \rangle_{i \in \nat}$ in the graph of~$c$, there exists an infinitely progressing trace $\langle \check{\Gamma}_i \rangle_{i \geq N}$ for some $N \in \nat$.
 Let $\langle \varphi_i \rangle_{i \geq N}$ be
 the sequence  of formulas $\varphi_i \in \Gamma_i$ that include the mark in $\check{\Gamma}_i$.
 Then this sequence forms a $\nu$-thread 
 since the variable of the fixed-point operator $\nu x$ tracked in  $\langle \check{\Gamma}_i \rangle_{i \geq N}$ 
 is minimal 
 by the result~\cite[Lemma 7]{DBLP:conf/mfcs/Bruse14}.

 Conversely, 
 suppose that
 every infinite path in $\Pi$ has a $\nu$-thread.
 Let
 $\langle v_i \in C(\Gamma_i) \rangle_{i \in \nat}$ be an infinite path, 
 and let
 $\langle \varphi_i \rangle_{i \geq N}$ be a $\nu$-thread along this path.
 Let~$x$ be the minimal variable among variables unfolded infinitely often.
  By increasing \(N\) if necessary, we may assume that no fixed-point
binder strictly smaller than \(x\) is unfolded along the suffix.
 We now construct a trace along this suffix as follows:
for each
\(n\geq N\), choose a later position \(k\geq n\) at which \(x\)
is unfolded, and propagate the mark on that occurrence backwards from
\(k\) to \(n\) along the thread.  The absence of unfoldings of binders
smaller than \(x\) ensures that this backward propagation is compatible
with the trace structure.  Thus we obtain a sequence of marked
sequents
  $\langle \check{\Gamma}_n\rangle_{n\geq N}$
forming a trace along the path.  
Since \(x\) is unfolded infinitely often, the
trace is infinitely progressing.  Hence the path satisfies the GTC.
\end{proof}

The semantics $a\colon \prffunc{P}\Omega \to \Omega$
in $\Psh{\Jdg}$ is defined by 
\begin{align*}
\Omega(\Gamma) &\coloneqq \{*\}
\text{ if } \Gamma \text{ is valid in $\mathcal{K}$, and } \emptyset \text{ otherwise}, \\
a_\Gamma(r, \langle *\rangle_i) &\coloneqq *.
\end{align*}
This assignment is well-defined because each rule is sound:
whenever all premises of a rule are valid, its conclusion is valid.
Then a ca-morphism from~$c$ to~$a$ yields that
any sequent in $\Pi$ is valid.

\subsubsection{abstract proof system $P^\gamma$} \label{sec:mu_p_gamma}
Let \(P\) be the abstract proof system for the modal \(\mu\)-calculus
defined in \autoref{subsec:modal_p},
and $\gamma$ be a sufficiently large ordinal, in particular larger than the cardinality of $2^S$
where $S$ is the set of states of the given LTS $\mathcal{K}$.
The proof system $P^\gamma$
can be seen as a transfinite extension of the original proof system~$P$
obtained
by annotating $\nu$-operators with ordinal numbers.
Such an ordinal-annotated system is often employed to prove soundness or completeness~\cite[Section 5]{DBLP:journals/tcs/NiwinskiW96}.

A judgement $(\Gamma, f)  \in \OJdg{\gamma}$
in the proof system $P^\gamma$
represents a sequent $\Gamma$ with ordinal annotations $\vec{\beta}$ on the $\nu$-operators in $\Gamma$;
we write $\Gamma^{\vec{\beta}}$ for this.
A rule arising from $\nu$ can be represented as\footnote{For simplicity, we suppress arities indexed by $\vec{\beta}' \leq \vec{\beta}$. This does not affect the proof system because the validity of these instances is ensured via the assignment $\vec{\beta}$.}:
\begin{center}
 \AxiomC{$\{\Gamma, \varphi^{\vec{\beta}}[\nu^{\alpha'} x.\varphi^{\vec{\beta}}/x]\}_{\alpha' < \alpha}$}
 \UnaryInfC{$\Gamma, \nu^\alpha x.\varphi^{\vec{\beta}}$}
 \DisplayProof
\end{center}
This rule is based on the transfinite computation of greatest fixed points~\cite{cousot1979constructive}:
for a monotone function $G\colon 2^S \to 2^S$,
$G^\gamma (\top) = \mathrm{gfp}(G)$
where
the approximant $G^\alpha$ for $\alpha \leq \gamma$ is defined by
  $G^0(\top) \coloneqq S$,
  $G^{\alpha+1}(\top) \coloneqq G(G^\alpha(\top))$,
and, if $\alpha$ is a limit ordinal,
  $G^\alpha(\top) \coloneqq \bigcap_{\alpha'<\alpha} G^{\alpha'}(\top)$.
Equivalently, using the convention that the intersection over the empty
family is \(S\), this definition can be written compactly as
  $G^\alpha(\top)
  =
  \bigcap_{\alpha'<\alpha} G(G^{\alpha'}(\top))$.

\subsubsection{soundness}
The transfinite computation of greatest fixed points 
induces that the semantic algebra~$a_{\mathcal{K}}$ in $\Psh{\Jdg}$ lies in 
the image of $\Alg{u_*^\gamma}$.
For an assignment $f\colon \fml(\Gamma) \to \gamma$ and a subformula $\varphi$ of a formula in $\Gamma$,
we define $\llbracket (\varphi, f) \rrbracket^\mathcal{K}_{\rho}$
by following the  definition of $\llbracket \varphi \rrbracket^\mathcal{K}_\rho$,
except for the $\nu$-formulas, which are interpreted as
\[
  \llbracket (\nu x.\varphi, f) \rrbracket_{\rho}^\mathcal{K} \coloneqq 
  (\lambda v \in 2^S.~\llbracket (\varphi, f) \rrbracket_{\rho[x \mapsto v]}^\mathcal{K})^{f(\nu_\bullet x.\varphi)}(\top),
\]
where
$f(\nu_\bullet x.\varphi)$ is shorthand for the ordinal assigned by $f$ to 
the marked sequent $\check{\Gamma}$ obtained by marking the occurrence of $\nu x.\varphi$.

Define 
the semantics $a'_\mathcal{K}\colon \prffunc{P^\gamma}(\Omega'_{\mathcal{K}}) \to \Omega'_{\mathcal{K}}$ in $\Psh{\OJdg{\gamma}}$ by
\begin{align*}
\Omega'_{\mathcal{K}}((\Gamma, f)) &\coloneqq \{*\} \text{ if } (\Gamma, f) \text{  is valid in $\mathcal{K}$, and } \emptyset \text{ otherwise}, \\
(a'_{\mathcal{K}})_{(\Gamma, f)}((r, f), \langle * \rangle_i) &\coloneqq *.
\end{align*}
Here, 
validity of $(\Gamma, f)$ is 
defined analogously to validity of $\Gamma$.

Noting that $G^\gamma (\top) = \mathrm{gfp}(G)$,
it follows that
$a = \Alg{u^\gamma_*}(a')$.
Therefore, by \autoref{thm:soundness}, for a proof \(c\), there is a
unique ca-morphism from \(c\) to \(a_{\mathcal{K}}\), implying that every sequent
\(\Gamma\) appearing in \(c\) is valid in the fixed LTS \(\mathcal{K}\).  
Since
\(\mathcal{K}\) was arbitrary, every such sequent is valid in every LTS.

\subsection{For Non-Wellfounded Proofs in \(\mu\)-bicomplete Categories} \label{sec:santocanale}

We finally discuss 
a non-wellfounded variant of 
Santocanale's
cut-free circular proof system~\cite{DBLP:conf/fossacs/Santocanale02}.  
 Santocanale originally
formulates proofs as finite graphs, while here we present the
same rules in the form of possibly infinite derivation trees, in accordance
with our coalgebraic framework.  
This example illustrates that our
framework accommodates abstract proof systems whose category of judgements
\(\Jdg\) is non-discrete.

\subsubsection{the calculus in \cite{DBLP:conf/fossacs/Santocanale02}}

We briefly recall the calculus described in~\cite{DBLP:conf/fossacs/Santocanale02}, adapting some notions to the present paper.  
Let $\lambda$ be an infinite regular cardinal,
and
\(\mathcal C\) be a locally $\lambda$-presentable category.\footnote{%
  For this example, choose universes $\mathcal{U} \in \mathcal{V}$ such that $\mathcal{C}$ is $\mathcal{V}$-small and $\mathcal{U}$-locally small. Local presentability
  and the chain-convergence assumptions refer to $\mathcal{U}$; the abstract framework is applied with $\mathcal{V}$-small indexing
  categories and $\mathcal{V}$-sets.
}
We fix a signature $\mathcal{H}$ of function symbols and an interpretation~$I$ of $\mathcal{H}$,
that is, $I(H)\colon \mathcal{C}^n \to \mathcal{C}$ for each~$n$-arity function symbol $H \in \mathcal{H}$.
Later,
we will impose a chain-convergence assumption:
parameterized initial
algebras and final coalgebras required for the interpretation of terms are obtained by the corresponding initial and final chains.  
This assumption ensures that \(\mathcal C\) is \(\mu\)-bicomplete~\cite{DBLP:journals/ita/Santocanale02}.
We write
$T(\mathcal{C})$ for the collection of
terms built from objects of \(\mathcal C\), variables, finite products
\(\bigwedge\), finite coproducts \(\bigvee\), and function symbols in $\mathcal{H}$.

A \emph{directed system of equations} over \(\mathcal C\) is a finite family
of labelled equations of the form
\[
  (x =_{\epsilon_S(x)} q_S(x))_{x\in X_S}
\]
where 
$x$ is a variable,
$q_S(x)$ is a term,
and
\(\epsilon_S(x)\in\{\mu,\nu\}\), together with a subset
\(X_{0,S}\subseteq X_S\) of designated variables, such that the dependency
graph on \(X_S\), defined by \(x\to y\) when \(y\) occurs in \(q_S(x)\), is a
forest with back edges rooted at \(X_{0,S}\).  
We write $(X_S,q_S,\epsilon_S,X_{0,S})$ for such a directed system and \(\leq_S\) for the induced order on variables.

We shall use
two directed systems
  $S=(X_S,q_S,\epsilon_S,X_{0,S})$ and
  $T=(X_T,q_T,\epsilon_T,X_{0,T})$
for variables occurring on the left-hand and right-hand side of sequents, respectively.  
A sequent is a
pair of terms, written \(s\vdash t\).
For a set \(Z\) of free variables 
such that
$Z \cap X_S = \emptyset$,
each term \(s\) whose free variables are contained
in \(X_S\cup Z\) is interpreted as a functor
\[
  \llbracket s\rrbracket^Z_S\colon \mathcal{C}^Z \to \mathcal{C},
\]
by following the structure of the term in the usual way,
together with
an induction on the size of directed systems and on $<_S$.
A variable in~$Z$ is interpreted as the projection on the variable,
while a variable in $X_S$ is determined as follows.

For a variable $x \in X_S$,
the interpretation of the term $q_S(x)$ under 
 induces the functor \[\llbracket q_S(x) \rrbracket^{Z \cup x^\downarrow}_{S_{> x}}\colon \mathcal{C}^{Z} \times \mathcal{C}^{(x^\downarrow)} \to \mathcal{C}\] 
 under a directed system $S_{> x}$ consisting of variables greater than~$x$,
 where $x^\downarrow \coloneqq \{y \mid y \leq_S x\}$.
Substituting the interpretations for variables strictly under~$x$,
we obtain a parameterized endofunctor $\Phi_x^Z\colon \mathcal{C}^Z \times \mathcal{C} \to \mathcal{C}$.
We assume that, for each parameter $\rho \in \mathcal{C}^Z$,
each endofunctor $\Phi_x^Z(\rho, \_)$ 
has its initial algebra and final coalgebra obtained as the colimit of the initial $\lambda$-chain and the limit of its final $\lambda^\op$-chain, respectively;
see \cite{DBLP:journals/jlp/AdamekMM18} for background on fixed points.
If \(\epsilon_S(x)=\mu\) (resp.~$\epsilon_S(x) = \nu$), then \(\llbracket x\rrbracket_S^Z\) is defined as
the carrier of the parameterized initial algebra (resp.~final coalgebra) of \(\Phi_x^Z\).  
We use the analogous notation
  $\llbracket t\rrbracket^W_T$
for a term \(t\) interpreted relative to the right-hand system \(T\).

We work with the assumption-free fragment, omitting the rule \(A\) in \cite[Sec.~2.2]{DBLP:conf/fossacs/Santocanale02}. 
The calculus then contains rules for constant morphisms in $\mathcal{C}$, functoriality for function symbols, finite products, finite coproducts, and fixed-point unfolding.
We do not spell out these standard rules here; the fixed-point unfolding rules are as follows:
\[
\begin{array}{c@{\qquad}c}
  \dfrac{q_S(x)\vdash t}
        {x\vdash t}
        \ L\epsilon_S(x) x
&
  \dfrac{s\vdash q_T(y)}
        {s\vdash y}
        \ R\epsilon_T(y) y,
\end{array}
\]
  for \(x\in X_S\) and \(y\in X_T\).
For example, if \(\epsilon_S(x)=\mu\), the first rule is written
\(L\mu x\).

\subsubsection{abstract proof system~$P$} \label{subsec:santocanale_p}
Let us define an abstract proof system~$P$ for the calculus above.
We fix sets
\(Z\) and \(W\) of free variables for the left- and right-hand sides,
respectively.
For a left term \(s\), its interpretation is a functor
  $\llbracket s\rrbracket_S:\mathcal{C}^{Z}\to\mathcal{C}$,
and for a right term \(t\), its interpretation is a functor
  $\llbracket t\rrbracket_T:\mathcal{C}^{W}\to\mathcal{C}$.

We do not take judgements to be merely syntactic sequents, since open terms
depend on assignments of their free variables.  Instead, following the
convention that semantic objects are presheaves on \(\Jdg\), we define
\newcommand{\Seq}{\mathrm{Seq}}
\[
  \Jdg
  \coloneqq
  \Seq_{S,T}\times \mathcal{C}^{Z}\times (\mathcal{C}^{W})^\op,
\]
where \(\Seq_{S,T}\) is the discrete category of sequents
\(s\vdash t\).  We write an object of \(\Jdg\) as
  $(s\vdash t;\rho,\sigma)$
where \(\rho\in\mathcal{C}^{Z}\) and \(\sigma\in\mathcal{C}^{W}\).  
The category $\Psh{\Jdg}$ can be seen as $(\Set^{(\mathcal{C}^{Z})^\op \times \mathcal{C}^{W}})^{\Seq_{S, T}}$.

We now define an abstract proof system~$P$
from the rules of the calculus.
As in \autoref{sec:modal},
the rules yield an abstract proof system $P_0 = (\Seq_{S, T} \xleftarrow{\prem_0} \GopfibTotal{\arfunc_0} \xrightarrow{\Gopfib{\arfunc_0}} R_0 \xrightarrow{\ccl_0} \Seq_{S, T})$.
Put $E\coloneqq\mathcal{C}^{Z}\times(\mathcal{C}^{W})^\op$.
We define
$R\coloneqq R_0\times E$ and
$\ccl(r,e)\coloneqq(\ccl_0(r),e)$.
For each $r\in R_0$ and $e\in E$, let
  $\arfunc(r,e)\coloneqq
  \coprod_{i\in\arfunc_0(r)} E/e$,
and 
for a morphism $k\colon e\to e'$ in $E$, we define
$\arfunc(\id_r,k)(i,h)\coloneqq(i,k\circ h)$.
The premise functor is defined by
  $\prem\big((i,h\colon e_0\to e)\in\arfunc(r,e)\big)
  \coloneqq(\prem_0(r,i),e_0)$.
  It maps
a morphism in $\GopfibTotal{\arfunc}$ above
$(\id_r,k):(r,e)\to(r,e')$, which is 
a morphism $a:e_0\to e_1$ s.t.~$h'\circ a=k\circ h$,
to
$(\id_{\prem_0(r,i)},a)$.
These data define an abstract proof system
  $P=\big(
  \Jdg\xleftarrow{\prem}\GopfibTotal{\arfunc}
  \xrightarrow{\Gopfib{\arfunc}}R
  \xrightarrow{\ccl}\Jdg
  \big)$.

For $G\in\Psh{\Jdg}$, its familial functor satisfies
\begin{align*}
  \prffunc{P}(G)(s,e)
  &\cong
  \coprod_{r\in(R_0)_s}
  \prod_{i\in\arfunc_0(r)}
  \lim_{(h:e_0\to e)\in(E/e)^\op}
  G(\prem_0(r,i),e_0)\\
  &\cong
  \coprod_{r\in(R_0)_s}
  \prod_{i\in\arfunc_0(r)}
  G(\prem_0(r,i),e).
\end{align*}
The second isomorphism is given by evaluation at
$h=\id_e$ because the object $\id_e$ is terminal in $E/e$,
and hence initial in $(E/e)^\op$.
These isomorphisms are natural in $G$ and $(s,e)$.

Then a derivation tree can be represented as a pre-proof $c\colon C \to \prffunc{P}(C)$ in $\Psh{\Jdg}$,
where $C(s \vdash t; \rho, \sigma)$ is the set of vertices labelled with $s \vdash t$.

We next define a trace structure.  Let
\[
  D_{S,T}\coloneqq
  \{\,L_x \mid x\in X_S,\ \epsilon_S(x)=\mu\,\}
  \cup
  \{\,R_y \mid y\in X_T,\ \epsilon_T(y)=\nu\,\}.
\]
The marker \(L_x\) means that we are following a left trace whose least
recurring left $\mu$-variable is intended to be \(x\).  Dually, \(R_y\)
means that we are following a right trace whose least recurring right $\nu$-variable is intended to be \(y\).  We take
  $\mathrm{fml}\in \widehat{\Jdg}$
to be the constant presheaf at \(D_{S,T}\).
A trace structure for~$P$ is  defined as follows.  
A marker $L_x$ propagates along every premise position
$(i,h)\in\arfunc(r,(\rho,\sigma))$ unless $r$ is a left
fixed-point rule for some $z\in X_S$ with $z<_S x$;
the transition is progressing exactly when $r=L\mu x$.
Dually, a marker $R_y$ propagates along every such premise
position unless $r$ is a right fixed-point rule for some
$z\in X_T$ with $z<_T y$;
the transition is progressing exactly when $r=R\nu y$.

The resulting GTC says that, along every infinite path, either a left $\mu$ or right $\nu$-variable is unfolded infinitely often, and this variable is minimal among the variables unfolded infinitely often on the corresponding side.
The original condition in \cite{DBLP:conf/fossacs/Santocanale02} is formulated as a condition on cycles because they work with circular proofs rather than non-wellfounded proofs.

Define a presheaf
\[
  \Omega(s\vdash t; \rho, \sigma)
  \coloneqq
  \mathcal{C}(\llbracket s\rrbracket_S^{Z}(\rho),\llbracket t\rrbracket_T^{W}(\sigma)).
\]
The interpretation of the rules induces an \(F[P]\)-algebra
  $a\colon F[P](\Omega)\to \Omega$
via the natural isomorphism above.

\subsubsection{abstract proof system \(P^{\gamma}\)}
Let~$P$ be the abstract proof system defined in \autoref{subsec:santocanale_p}.
Recall that we assume that the semantics of variables in \(X_S\) and \(X_T\)
is obtained by ordinal approximation: for each equation, the associated
parameterized endofunctor has its initial algebra or final coalgebra obtained
as the colimit of the initial \(\lambda\)-chain or the limit of the final \(\lambda^{\op}\)-chain, respectively.
We put $\gamma \coloneqq \lambda+1$.

Applying the construction in \autoref{def:ojdg} to the trace
structure above, we obtain an ordinal-annotated proof system
  $P^\gamma$.
A judgement of \(P^\gamma\) is a pair
  $((s\vdash t;\rho,\sigma),f)$,
where
  $f\colon D_{S, T}\to \gamma$
assigns an ordinal to each left $\mu$ or right $\nu$-variables.  

\subsubsection{Soundness}

We define a presheaf $\Omega' \in \Psh{\Jdg^\gamma}$ by
\[
  \Omega'((s\vdash t;\rho,\sigma),f)
  \coloneqq
  \mathcal{C}
  \bigl(
    \llbracket s\rrbracket_{S,f}^{Z}(\rho),
    \llbracket t\rrbracket_{T,f}^{W}(\sigma)
  \bigr),
\]
where
  $\llbracket s\rrbracket^{Z}_{S,f}\colon \mathcal C^{Z}\to\mathcal C$ and
  $\llbracket t\rrbracket^{W}_{T,f}\colon \mathcal C^{W}\to\mathcal C$
  are defined
  in almost the same way as $\llbracket s\rrbracket^{Z}_{S}$ and $\llbracket t\rrbracket^{W}_{T}$, respectively,
except for bound variables.  
For \(x\in X_S\), let
  $\Phi^{Z}_{x,f}\colon \mathcal C^{Z}\times\mathcal C\to\mathcal C$
be the parameterized endofunctor 
defined by $\Phi^{Z}_{x,f}(\rho, c) \coloneqq \llbracket q_S(x) \rrbracket_{S_{>x}, \Delta_\lambda}^{Z \cup x^\downarrow}[\llbracket y \rrbracket^{Z}_{S, f}/y, c/x]_{y <_S x}(\rho)
= \llbracket q_S(x) \rrbracket_{S_{>x}, \Delta_\lambda}^{Z \cup x^\downarrow}(\rho, (\llbracket y \rrbracket^{Z}_{S, f})_{y < x}, c)$.
Then for \(\rho\in\mathcal C^{Z}\), we define
$\llbracket x\rrbracket^{Z}_{S,f}(\rho)$ to be the $f(L_x)$-th object of the initial chain of $\Phi^{Z}_{x,f}(\rho, \_)$ if $\epsilon_S(x) = \mu$,
and the carrier of the final coalgebra of $\Phi^{Z}_{x,f}(\rho, \_)$ otherwise.
Its action on morphisms is defined by the same transfinite induction, using functoriality of $\Phi^Z_{x, f}$ and the universal property of colimits.
Dually, \(\llbracket t\rrbracket^{W}_{T,f}\) is defined in the same way for
the right-hand system \(T\), using \(f(R_y)\) as the approximation stage for
each variable \(y\in X_T\) with $\epsilon_T(y) = \nu$.

In what follows, we focus on the left-hand side.
The corresponding definitions and arguments for the right-hand side are obtained dually.
For each $f\colon D_{S, T} \to \gamma$, $x \in X_S$, and $\beta \in \gamma$,
we write $f[x \mapsto \beta]$ for 
the function mapping~$d$ to $\beta$ if $\epsilon_S(x) = \mu$ and $d = L_x$,
and $f(d)$ otherwise.
\begin{lem} \label{lem:santocanale_sem}
  For each $f\colon D_{S, T} \to \gamma$, the following statements hold.
  \begin{enumerate}
    \item \label{item:y_gam} For each bound variable $x\in X_S$,
      $\llbracket x \rrbracket_{S, f}^Z = \llbracket x \rrbracket_{S, f[y \mapsto \lambda]_{y >_S x}}^Z$ holds,
      and it further gives that
      $\Phi_{x, f}^Z = \Phi_{x, f[x \mapsto \beta]}^Z$ for each $\beta \in \gamma$.
    \item \label{item:substi} For each (left) term~$s$ whose variables from $X_S$ are in $\{y \mid y \leq_S x \text{ or }y \geq_S x\}$, $\llbracket s \rrbracket^Z_{S, f} \cong \llbracket s\rrbracket^{Z \cup x^\downarrow}_{S_{> x}, f}[\llbracket y \rrbracket^Z_{S, f}/y]_{y \leq_S x}$.
    \item  \label{item:s_gamma}
    For each (left) term~$s$,
    $\llbracket s \rrbracket^Z_{S, \Delta_\lambda} \cong \llbracket s \rrbracket^Z_{S}$.
  \end{enumerate}
\end{lem}
\begin{proof}[Proof sketch]
The proof is by the same induction as the definition of the semantics:
structural induction on terms, together with well-founded induction on
the lexicographic order consisting of the size of the directed system and
the strict order \(<_S\).  
The last statement
\eqref{item:s_gamma} follows from the chain-convergence assumption: at the
top assignment \(\Delta_\lambda\), the relevant ordinal approximants have
converged to the initial algebras or final coalgebras used in the
ordinary semantics.
\end{proof}

We now define an algebra
  $a'\colon
  F[P^{\gamma}](\Omega')
  \to
  \Omega'$ by interpreting the rules of $P^\gamma$.
  For rules other than fixed-point unfolding, the interpretation is given by the same way as in $a$.
  For the rule $L\mu x$ with $f\colon D_{S, T} \to \gamma$,
  by \autoref{lem:santocanale_sem}.\ref{item:y_gam},
  $\llbracket x \rrbracket_{S, f}^Z(\rho)$ is equal to the $f(L_x)$-th object of the initial chain of $\Phi_{x, f[x \mapsto \lambda]}^Z(\rho, \_)$.
  \autoref{lem:santocanale_sem}.\ref{item:substi} and \autoref{lem:santocanale_sem}.\ref{item:y_gam} induce that
  $\llbracket x \rrbracket_{S, f}^Z(\rho)$
  is isomorphic to 
  $\colim_{\beta < f(L_x)} \llbracket q_S(x)\rrbracket^Z_{S, f[x \mapsto \beta, z \mapsto \lambda]_{z >_S x}}(\rho)$
  because
  \begin{align*}
  &\llbracket q_S(x)\rrbracket^Z_{S, f[x \mapsto \beta, z \mapsto \lambda]_{z >_S x}}(\rho) \\
  &\cong \llbracket q_S(x)\rrbracket^{Z \cup x^\downarrow}_{S_{> x}, f[x \mapsto \beta, z \mapsto \lambda]_{z >_S x}}[\llbracket y \rrbracket^Z_{S, f[x \mapsto \beta, z \mapsto \lambda]_{z >_S x}}/y]_{y \leq_S x}(\rho) &\text{by }\autoref{lem:santocanale_sem}.\ref{item:substi} \\
  &\cong \llbracket q_S(x)\rrbracket^{Z \cup x^\downarrow}_{S_{> x}, \Delta_\lambda}[\llbracket y \rrbracket^Z_{S, f[x \mapsto \beta]}/y]_{y \leq_S x}(\rho) &\text{by }\autoref{lem:santocanale_sem}.\ref{item:y_gam} \\
  &\cong \Phi_{x, f}^Z(\rho, \llbracket x \rrbracket^Z_{S, f[x \mapsto \beta]}(\rho))  &\text{by }\autoref{lem:santocanale_sem}.\ref{item:y_gam}.
  \end{align*}
  Then the semantics for this rule is given by the universality of the colimit.

The construction of \(a'\) together with \autoref{lem:santocanale_sem}.\ref{item:s_gamma}
yields
  $a \cong \Alg{u^\gamma_*}(a')$.
Therefore \autoref{thm:soundness} provides that the proof system~$P$ is sound.

\section{Related Work} \label{sec:relatedwork}

\hspace*{\parindent}
\emph{a) Abstract cyclic proofs}:
As discussed in the introduction,
Afshari and Wehr introduced \emph{abstract cyclic proofs}~\cite{AfshariW22},
where the GTC is described in a categorical manner.
Building on their framework, Leigh and Wehr studied a proof translation from cyclic proofs with GTCs to those with local conditions called \emph{reset conditions},
aiming at more efficient validity checking of cyclic proofs~\cite{DBLP:journals/apal/LeighW24}.
Their work focuses primarily on algorithmic aspects,
and graph structures of paths and traces are represented using semi-categories and relations.
In contrast, we basically represent graphs as coalgebras of generalised polynomial functors.
This is 
motivated by our goal of establishing a soundness result within an abstract framework,
namely,
interpreting proofs as ca-morphisms.

\emph{b) Realisations}:
Kozen introduced \emph{realisations} as a way to generalise polynomial functors on $\Set$ 
to those with labels, representing type signatures as directed multigraphs
\cite{DBLP:journals/entcs/Kozen11}.
Rather than specifying graph structures by such graph formalisms,
we derive them as graphs of coalgebras in presheaves, following a natural extension of the construction on $\Set$ developed in~\cite{AdamekMM20}.
Jeannin et al.~\cite{DBLP:journals/mscs/JeanninKS17}
established an analogous result to \autoref{thm:wf} for functors presented via realisations,
whereas our approach applies more generally to functors on presheaves that are not necessarily polynomial.

\emph{c) Abstract ordinal trace arguments}:
Brotherston's general soundness result~\cite[Sec.~4.2.1]{brotherston2006sequent} uses ordinal trace values compatible with
the propagation of countermodels.
Our use of ordinal annotations follows the standard
ordinal descent approach to soundness, which also underlies
his general soundness theorem.
Here the ordinal construction yields a recursive coalgebra, and an
adjunction then transfers its unique ca-morphisms to the original
system. 
Also, our framework allows non-discrete categories of judgements, which is used in particular, in our treatment of
ordinal annotations (\autoref{sec:ord_system}) and of open terms
interpreted relative to variable assignments
(\autoref{sec:santocanale}).

\emph{d) Categorical global conditions in automata theory}:
There is a line of work that studies global conditions, such as B\"uchi and parity conditions,
from a categorical and coalgebraic perspective, particularly in automata theory~\cite{DBLP:conf/cmcs/UrabeH18,DBLP:conf/cmcs/CianciaV12}.
These approaches aim at providing a categorical formalisation of accepted behaviours.
While our GTC can be seen as a generalisation of B\"uchi conditions,
it serves a different purpose:
it is formulated as a property of a coalgebra itself to guarantee the existence of a ca-morphism.

\section{Conclusion} \label{sec:conclusion}
To establish a characterisation of the GTC in terms of recursiveness,
we developed a coalgebraic framework for non-wellfounded proofs by exploiting graphs of coalgebras in presheaves and the correspondence of ca-morphisms along adjunctions.
Within this framework, we established a soundness theorem:
any pre-proof $c$ satisfying
the GTC has a unique solution with respect to 
$\polyfunc{P}$-algebra $\Alg{u^\gamma_*}(a')$, for each $\polyfunc{P^\gamma}$-algebra $a'$.
The solution is a ca-morphism for the original proof system $P$,
obtained by transposing the unique solution of the annotated coalgebra.
We instantiated this construction for the modal $\mu$-calculus and
the non-wellfounded variant of Santocanale's circular calculus.
We also studied a relationship between recursiveness of coalgebras and the GTC, including a coalgebraic characterisation of GTC in terms of recursive coalgebras.

In future work,
we plan to study further properties of proof systems,
such as completeness and cut-elimination, 
in our coalgebraic framework.
Another direction of future work is to formulate other soundness conditions for non-wellfounded proofs, such as reset conditions~\cite{DBLP:conf/lics/AfshariL17,DBLP:journals/apal/LeighW24} and bouncing threads~\cite{10.1145/3531130.3533375},
and to investigate the relationships between proof systems with these conditions and their
corresponding soundness theorems.

\section*{Acknowledgment}
  \noindent The author would like to thank Shin-ya Katsumata and Keisuke Hoshino for helpful discussions.
  The author is supported by JST ACT-X, Grant  No.~JPMJAX25CD.

\bibliographystyle{alphaurl}
\bibliography{mybib}

\appendix

\section{Omitted Proofs}
\subsection{Proof of \autoref{lem:coalg_to_alg}}
\begin{proof}
  Let~$c$ be an~$F$-coalgebra $C \to FC$ and~$a$ be an~$G$-algebra $GA \to A$.
  Here we shall write $\ol{(-)}$ for the adjoint transposition.

  Then
  a morphism~$i$ from~$c$ to $\Alg{R}(a)$  
  is a morphism $i\colon C \to RA$ such that
  $i = (C \xxrightarrow{c} FC \xxrightarrow{Fi} FRA \xxrightarrow{\sigma_A} RGA \xxrightarrow{Ra} RA)$,
  and
  $i'$ from $\Coalg{L}(c)$ to~$a$ is a morphism $i'\colon LC \to A$ such that $i' = (LC \xxrightarrow{Lc} LFC \xxrightarrow{\tau_C} GLC \xxrightarrow{Gi'} GA \xxrightarrow{a} A)$.
  These morphisms~$i$ and $i'$ correspond under the adjoint transposition:
  because of the correspondence between $\tau$ and $\sigma$,
  we have
  \begin{align*}
    i &= (C \xxrightarrow{c} FC \xxrightarrow{Fi} FRA \xxrightarrow{\sigma_A} RGA \xxrightarrow{Ra} RA)\\
    &= (C \xxrightarrow{c} FC \xxrightarrow{\ol{\tau_C}} RGLC \xxrightarrow{RG\ol{i}} RGA \xxrightarrow{Ra} RA) \\
    &= \left(\ol{LC \xxrightarrow{Lc} LFC \xxrightarrow{\tau_c} GLC \xxrightarrow{G\ol{i}} GA \xxrightarrow{a} A}\right).
  \end{align*}
  The second equality holds because
  $(LFC \xxrightarrow{\tau_C} GLC \xxrightarrow{G\ol{i}} GA)
  = (LFC \xxrightarrow{\tau_C} GLC \xxrightarrow{GLi} GLRA \xxrightarrow{G\epsilon_A} GA)
  = (LFC \xxrightarrow{LFi} LFRA \xxrightarrow{\tau_{RA}} GLRA \xxrightarrow{G\epsilon_A} GA)$
  shows $RG\ol{i} \circ \ol{\tau_C} = RG\epsilon_A \circ \ol{\tau_{RA}} \circ Fi = \sigma_A \circ Fi$.
\end{proof}

\subsection{Proof of \autoref{prop:vpoly_morph}} \label{ap:vpoly_morph}
Let us introduce the following two facts before going into the proof.
\begin{prop} \label{prop:equiv_ra}
  There is an equivalence between $[\cat{B}, \Psh{\cat{A}}]^\op$ and the category of right adjoints $\Psh{\cat{A}} \to \Psh{\cat{B}}$ and natural transformations between them.
  It is given by the nerve construction: 
  it maps a functor $F\colon \cat{B} \to \Psh{\cat{A}}$ to the right adjoint $R\colon \Psh{\cat{A}} \to \Psh{\cat{B}}$ defined by $R(X)(b) = \Psh{\cat{A}}(Fb, X)$ for each $X \in \Psh{\cat{A}}$ and $b \in \cat{B}$,
  and it maps a right adjoint $R\colon \Psh{\cat{A}} \to \Psh{\cat{B}}$ to $L \circ y\colon \cat{B} \to \Psh{\cat{A}}$
  where~$L$ is the left adjoint of~$R$ and $y\colon \cat{B} \to \Psh{\cat{B}}$ is the Yoneda embedding.
\end{prop}
\begin{prop} \label{prop:corresp_natural}
  Let $f\colon \cat{A} \to \cat{A}'$, $g\colon \cat{B} \to \cat{B}'$,
  $P\colon \cat{A}^\op \times \cat{B} \to \Set$, $P'\colon \cat{A'}^\op \times \cat{B'} \to \Set$ be functors.
  There is a one-to-one correspondence among the following data:
  \begin{itemize}
    \item a natural transformation $\tau\colon P \Rightarrow P' \circ (f^\op \times g)\colon \cat{A}^\op \times \cat{B} \to \Set$,
    \item a natural transformation $\sigma\colon \ol{P} \Rightarrow f^* \circ \ol{P'} \circ g\colon \cat{B} \to \Psh{\cat{A}}$ where $\ol{P}\colon \cat{B} \to \Psh{A}$ and $\ol{P'}\colon \cat{B}' \to \Psh{\cat{A'}}$ are given by currying~$P$ and $P'$, respectively,
    \item a natural transformation $\rho\colon g^* \circ P'^\dagger \Rightarrow P^\dagger  \circ f^*\colon \Psh{\cat{A'}} \to \Psh{\cat{B}}$ where $P^\dagger\colon \Psh{\cat{A}} \to \Psh{\cat{B}}$ is given by $P^\dagger(X)(b) = \Psh{\cat{A}}(\ol{P}b, X)$ and $P'^\dagger$ is defined similarly.
  \end{itemize}
\end{prop}
\begin{proof}[Proof of \autoref{prop:vpoly_morph}]
1) Because right adjoints preserve $1$,
the middle square of Diagram~\ref{eq:poly_morph_nat} induces an isomorphism at $1$.
Therefore, $\tau_1$ is an isomorphism.

2)
We briefly sketch the construction of Diagram~\ref{eq:poly_morph} from $\tau$.
Since $\tau_1$ is an isomorphism,
$\tau$ can be decomposed as the following left diagram by naturality of $\tau$.
\[
\xymatrix@R=1.7em@C=1.7em{
  \widehat{\mathbb{A}'} 
    \ar[r]^-{F'_1} &\widehat{\mathbb{A}'}/F'1 
    \ar[r] &\widehat{\mathbb{A}'} \\
  \widehat{\mathbb{A}} 
    \ar[u]^{u^*}
    \ar[r]_-{F_1} &\widehat{\mathbb{A}}/F1 \ar@{=>}[ul]^\tau
    \ar[u]_{u^*}
    \ar[r] &\widehat{\mathbb{A}}
    \ar[u]_{u^*}
} \quad 
\xymatrix@R=1.7em@C=1.7em{
  \widehat{\mathbb{A}'}/F'1 \ar[r]_{\simeq}
    &\Psh{\mathbb{A}'/F'1} 
     \\
  \widehat{\mathbb{A}}/F1 \ar[u]_{u^*}  \ar[r]_{\simeq}
    &\PshNonSmash{\mathbb{A}/F1} \ar[u]^{v^*}
} 
\]
Let 
$v\colon \cat{A}'/F'1 \to \cat{A}/F1$ be the functor defined by $v(a', x) \coloneqq (ua', \tau^{-1}(x))$ for each $a' \in  \cat{A}'$ and $x \in (F'1)(a')$, and $v(f) \coloneqq uf$ for each morphism~$f$.
By \autoref{prop:equiv_ra} and \autoref{prop:corresp_natural},
the natural transformation $\tau\colon u^* \circ F_1 \Rightarrow F'_1 \circ u^*$
with the right above diagram yields
a corresponding natural transformation $\alpha$ as in the following diagram.
\[
\xymatrix@R=1.7em@C=1.7em{
  \widehat{\mathbb{A}'} 
    &\mathbb{A}'/F'1 \ar[l]^{\,}="top"
    \ar[d]^{v} 
    \ar[r]  \pullbackmark{0, 1}{1, 0} &\mathbb{A}' 
    \ar[d]_{u} \\
  \widehat{\mathbb{A}} 
    \ar[u]^{u^*}
    &\mathbb{A}/F1 \ar[l]_{\,}="bot"
    \ar[r] &\mathbb{A}
    \ar@{=>}^\alpha"top"+<0pt,-4pt>; "bot"+<0pt,4pt>
}
\]
By \autoref{prop:equiv_ra},
the top and bottom horizontal functors above are 
functors appeared in the construction of polynomials corresponding to $F'$ and~$F$, respectively, as given after \autoref{prop:pra_poly}.
The natural transformation $\alpha$ 
corresponds to a morphism between profunctors, 
which further corresponds to
a morphism 
between two-sided discrete fibrations~\cite[Thm.~25]{ROSEBRUGH1988271}.
\end{proof}
\end{document}